\RequirePackage{fix-cm}
\documentclass[smallcondensed]{svjour3}     % onecolumn (ditto)
\smartqed  % flush right qed marks, e.g. at end of proof
\usepackage{natbib}
\usepackage{graphicx}
\usepackage{balance}  % for  \balance command ON LAST PAGE  (only there!)

\usepackage{amsmath}
\usepackage{mathrsfs}
\usepackage[lined,boxed,commentsnumbered,ruled]{algorithm2e}
\usepackage{algorithmic}
\usepackage{array}
\usepackage{subfigure}
\usepackage{color}
\usepackage{url}
\usepackage{multirow}
\usepackage{amssymb}
\usepackage{eucal}
\usepackage{mathptmx}      % use Times fonts if available on your TeX system

\newcommand{\mbl}[1]{\textcolor{blue}{#1}}
\newcommand{\nop}[1]{}
\newcommand{\todo}[1]{\textcolor{red}{#1}}
\begin{document}

\title{Mining Top-k Sequential Patterns in Database Graphs%\thanks{Grants or other notes
}
\subtitle{A New Challenging Problem and a Sampling-based Approach}



\author{Mingtao Lei  \and Lingyang Chu \and Zhefeng Wang
}


\institute{Mingtao Lei \at
                Key Laboratory of Trustworthy Distributed Computing and Service (BUPT), Ministry of Education \& School of Cyberspace Security, Beijing University of Posts and Telecommunications, Beijing, China \\
                \email{leimingtao@bupt.edu.cn}           %  \\
                \and
               Lingyang Chu \at
               Simon Fraser University, Burnaby, Canada \\
               \email{lca117@sfu.ca}
               \and
               Zhefeng Wang \at
               University of Science and Technology of China, Hefei, China\\
               \email{zhefwang@mail.ustc.edu.cn}
               }

\date{Received: date / Accepted: date}

\maketitle

\begin{abstract}
%%\mbl{
In many real world networks, a vertex is usually associated with a transaction database that comprehensively describes the behaviour of the vertex.
A typical example is the social network, where the behaviour of every user is depicted by a transaction database that stores his daily posted contents.
A transaction database is a set of transactions, where a transaction is a set of items. Every path of the network is a sequence of vertices that induces multiple sequences of transactions. The sequences of transactions induced by all of the paths in the network forms an extremely large sequence database.
Finding frequent sequential patterns from such sequence database discovers interesting subsequences that frequently appear in many paths of the network.
However, it is a challenging task, since the sequence database induced by a database graph is too large to be explicitly induced and stored.
In this paper, we propose the novel notion of database graph, which naturally models a wide spectrum of real world networks by associating each vertex with a transaction database.
Our goal is to find the top-$k$ frequent sequential patterns in the sequence database induced from a database graph.
We prove that this problem is \#P-hard. 
To tackle this problem, we propose an efficient two-step sampling algorithm that approximates the top-$k$ frequent sequential patterns with provable quality guarantee.
Extensive experimental results on synthetic and real-world data sets demonstrate the effectiveness and efficiency of our method.
%%}
\nop{
Mining sequential patterns in such networks may lead to informative insights.  
However, to the best of our knowledge, sequential pattern mining in such a context has never been tackled in literature.  }

\keywords{Database Graph \and Sequential Pattern Mining \and Uniform Sampling}
\end{abstract}

\section{Introduction}
\label{sec:intro}

Graphs and networks are popularly used to model advanced applications, such as social network analysis and communication network fault detection.  More often than not, rich data exists in such graph and network applications.  Consequently, graphs and networks in such applications have to be enriched by capturing more information in vertices and edges, such as labeled graphs~\citep{Ye2017} and attributed graphs~\citep{Pfeiffer2014}, where in a labeled graph each vertex is associated with a unique label, and in an attributed graph each vertex is associated with an attribute vector.  Mining various types of patterns in labeled graphs and attributed graphs has enjoyed interesting applications~\citep{Dutta2017,Tong2007}.

In some advanced graph applications, a vertex may contain much more information than just a label or an attribute vector. Such rich information on each vertex more often than not is better captured by a transaction database.
For example, in content-rich social networks, such as Twitter, YouTube, WeChat and DBLP, a vertex modelling a user can be associated with a transaction database that stores multiple transactions of content, where each transaction stores a tweet, a video clip, a post or a publication. 
As another example, in a road network of point of interests (POIs), each vertex representing a POI can be associated with a transaction database of visitor comments, where each transaction stores the keywords of one visitor comment.

To better model such graphs with rich information in vertices, we propose a novel notion of \emph{database graph}, where each vertex is associated with a transaction database.
Comparing with the label and the attribute vector on each vertex, the transaction database associated with each vertex naturally and concisely captures much more valuable information, such as the co-occurrences of items and the frequencies of patterns (i.e., set of items).

\nop{To the best of our knowledge, mining frequent sequential patterns from graphs where vertices contain rich information has not been tackled in literature. }

Finding interesting patterns in a database graph with rich content in vertices may lead to informative analytic results. For example, by mining sequential patterns (i.e., frequent subsequences) from all possible random walks in a social network with rich content in vertices, one may find interesting interaction patterns among users. Specifically, consider a social network of users, where each vertex, as a user, is associated with a transaction database storing the daily posts composed by the user, and each transaction stores a set of keywords of one post. In this database graph, a sequential pattern $\langle (\text{AI, Alpha Go}) \rightarrow (\text{deep learning, deepmind}) \rightarrow (\text{AI, medical doctor})\rangle$ indicates that it happens frequently that a user who writes about AI and Alpha Go connects to another user who writes about deep learning and deep mind, and further connects to a third user who writes about AI and medical doctor. It is likely that those topics may stimulate one and another.

The task illustrated in the above example is very different from traditional sequential pattern mining. Specifically, instead of searching for patterns from a sequence database, here we are given a database graph, where each vertex is associated with a transaction database, and our goal is to find sequential patterns that are frequently induced by all possible random walks of the whole graph. \nop{\todo{[What do the following two sentences mean? We would like to show that, our patterns actually represent the relationships among vertices and the random walks are quite different from the sequences defined in the sequence databases.]} \mbl{We demand for that, if a sequential pattern is induced by a random walk, every vertex in this random walk must donate at least one item for the pattern. This constraint on patterns leads to a different usage conditions for pruning strategies. }}

As we want to find sequential patterns, a natural question is whether the many existing sequential pattern mining methods, such as PrefixSpan~\citep{Pei2001}, can be extended to solve this problem? Unfortunately, there is no a straightforward way to apply existing methods on a database graph, since none of those methods take a database graph as input and the number of transaction sequences induced by all possible random walks of a database graph is exponential with respect to the number of vertices.

\nop{However, in these applications, networks can not be formulated as labeled graphs and attributed graphs. The reason is that, not only a label but also an attribute vector, can not express the co-occurrence of items and the pattern frequencies.}

To the best of our knowledge, our study is the first to tackle the problem of finding top-$k$ sequential patterns in database graphs. Our major idea is to maintain a sample of random transaction sequences so that we can approximate the top-$k$ patterns with provable quality guarantees. We make several contributions. 

First, we formulate the problem of mining top-$k$ sequential patterns in database graphs. A database graph in our problem is defined as a graph where each vertex is associated with a transaction database and a transaction is a set of items. Different from the traditional sequential pattern mining problem that aims to find frequent patterns in a given sequence database, we are interested in finding frequent sequential patterns in the extremely large sequence database that is induced by all possible random walks of the whole database graph. The major challenge of our problem is that the sequence database of a large database graph is usually too large to be explicitly induced and stored.
By reducing from the conventional sequential pattern mining problem, we prove that our problem is \#P-hard. 

Second, we propose an exact sequential pattern finding algorithm, which is also used as a baseline. By fixing the path length $l$ and an integer $k$, we first collect all of the length-$l$ transaction sequences and then apply one of the state-of-the-art sequential pattern mining techniques to obtain the exact top-$k$ patterns. However, when the size of the database graph increases, this approach suffers from the expensive time and space cost, since the total number of transaction sequences is exponential with respect to the number of vertices.

Third, we carefully design a two-step sampling framework that significantly improves the sequential pattern mining efficiency. 
We first sample a set of transaction sequences from the database graph by a two-step sampling algorithm. 
Using the sampled transaction sequences, we design an unbiased estimator to approach the frequencies of sequential patterns in the database graph with provable guarantee on the estimation error.

\nop{
build an unbiased pattern estimator that uses the sampled transaction sequences to estimate the pattern frequencies. 
In the first step, we sample a length-$l$ path $p=\langle v_1,\ldots, v_q, \ldots, v_{l+1} \rangle$ from the graph. Then in the second step, we sample a transaction sequence $ts$ from the sampled path $p$, where a transaction $T_q$ of $ts$ is sampled from the database $\mathcal{T}_{{v_q}}$ of vertex $v_q$ with the probability $\frac{1}{|\mathcal{T}_{{v_q}}|}$. Notice that, this two-step sampling is not uniform for sampling length-$l$ transaction sequences. We rectify the frequencies of sampled transaction sequences in the mining process following the unbiased pattern frequency estimator. That is, for a pattern $s$ contained in the transaction sequence sampled from $p^i$, we consider that $s$ appears $M^i$ times, where $M^i=\prod^{l+1}_{i=1}|\mathcal{T}_{v_q^i}|$ is the number of all transaction sequences supported by the $i$-th sampled length-$l$ path $p^i$. To achieve the quality guarantee, we present a bound on the sample size for a fixed $l$, which is $\frac{12|I|(l+1)+12}{\varepsilon^2 a}  ln\frac{2}{\delta}$, where $\varepsilon,\delta \in (0,1)$, $a=\frac{1}{|P_l|M^*} \sum_{i=1}^{|P_l|}M^i$ and $M^*=\max_i M^i$.}

Last, we conduct extensive experiments on both synthetic and real-world data sets. The results demonstrate the effectiveness and efficiency of the proposed method. We also conduct a case study to show meaningful sequential patterns mined from the ArnetMiner data set \citep{Tang2008}.

The rest of the paper is organized as follows. Section~\ref{sec:relatedwork} reviews the related work. We formulate the problem and present a baseline in Section~\ref{sec:prob}. In Section~\ref{sec:sampleAl}, we develop our two-step sampling method and give the upper bound of sample complexity for our proposed algorithm. A systematic empirical study is reported in Section~\ref{sec:exp}. We conclude our work in Section~\ref{sec:con}.

\section{Related Work}
\label{sec:relatedwork}

To the best of our knowledge, mining sequential patterns in database graphs is a new problem, which has not been touched in literature. 
It is related to sequential pattern mining and sampling methods.

\subsection{Sequential Pattern Mining}
Sequential pattern mining is a well studied subject in data mining, which was first introduced by~\cite{Agrawal1995}. It finds all frequent subsequences from a database of sequences and has enjoyed many applications, such as analyzing shopping patterns~\citep{Shang2016}, classification~\citep{Wang2016} and understanding user behavior~\citep{Zheng2016}. 
The Apriori-based algorithms, such as GSP \citep{Srikant1996} and SPADE \citep{Zaki2001}, improve the efficiency of sequential pattern mining~\citep{Agrawal1995} by reducing the search space using the anti-monotonicity of sequential patterns. Distributed methods were also proposed to accelerate the mining processing~\citep{Ge2016}.
To avoid generating candidate sequences, \cite{Pei2000} proposed FreeSpan. Later, \cite{Pei2001} developed PrefixSpan. The two methods mine sequential patterns by sequence database projection and pattern growth. 

Choosing an appropriate support threshold for sequential pattern mining is challenging in practice~\citep{Fournier2013}. \cite{Tzvetkov2005} proposed TSP, which aims to mine top-$k$ frequent closed sequential patterns whose lengths pass a threshold. Closed sequential patterns are concise representations of sequential patterns. 

\cite{Liu2014} proposed to proactively reduce the representation of sequences to uncover significant, hidden temporal structures. Their key idea was to decrease the level of granularity of symbols in sequences by clustering the symbols. They claimed that, in this way, symbolic sequences would be represented as numerical sequences or point clouds in the Euclidean space which reduces the time cost. However, this reduction does not change the hardness of the sequential pattern mining problem and may lose some interesting patterns due to the information loss caused by embedding.

There are many existing sequential pattern mining methods.  A thorough review on the subject is far beyond the capacity of this paper.  \cite{Dong2007} provided a comprehensive review.
Mining sequential patterns from database graphs is a new problem. As explained in Section~\ref{sec:intro}, the existing methods cannot be directly applied on a database graph.

\subsection{Sampling Methods}
Sampling and approximation methods~\citep{Thompson2012, Cochran1977} have been widely used in pattern mining to tackle large data sets or expedite mining. Sampling methods have been widely used in pattern mining~\citep{Calders2010} and association rule mining~\citep{Toivonen1996}, which reduce the amount of computation dramatically and achieve provable quality guarantees.

\cite{Rassi2007} proposed to ease the sequential pattern mining operations by modeling the data as a continuous and potentially infinite stream. They used Hoeffding concentration inequalities to prove a lower bound of the sample size. For reducing static database access by constructing a random sample before mining, they extended the static sampling analysis to the data stream model. They also conducted a simple replacement algorithm using an exponential bias function to regulate the sampling. Another study about static sampling for patterns was proposed by~\cite{Riondato2012}. They applied the statistical concept of \emph{Vapnik-Chervonenkis (VC) dimension} to develop a novel technique that provides tight bounds on the sample size. The resulting sample size was linearly dependent on the VC-dimension of a range space associated with the dataset to be mined. However, their method cannot be straightforwardly extended to sample transaction sequences from a database graph, since it is not applicable to deal with the transaction sequences.

Progressive sampling methods are also widely used. \cite{Pietracaprina2010} proposed to mine top-$k$ frequent itemsets through progressive sampling. They first presented an upper bound on the sample size. \nop{which guaranteed that the top-$k$ frequent itemsets mined from a random sample, with probability larger than a specified value.}Then they devised a progressive sampling approach that extracted the top-$k$ frequent itemsets from increasingly larger samples until suitable stopping conditions were met or the upper bound was hit. \cite{Riondato2015} introduced another progressive sampling method with \emph{Rademacher Averages} \todo{~\citep{Bartlett2002}} for mining frequent itemsets. This work studied the trade-off between approximation quality and sample size using concepts and results from \emph{statistical learning theory}~\citep{Cherkassky1997}.
The algorithm started from a small sample and enlarged it until a suitable stopping condition was met and a high-quality approximation can be obtained from the sample.
The stopping condition was based on bounds to the empirical Rademacher average of the problem. However, these methods focused on the conventional frequent itemset mining problem, but cannot be directly extended to find sequential patterns in database graphs.

Sampling transaction sequences from a database graph is also related to graph path sampling methods~\citep{Zhang2015, Ribeiro2012}.
For example,~\cite{Zhang2015} proposed Panther to measure similarity between vertices by sampling paths with random walk. However, Panther does not consider transaction databases on vertices, thus it cannot be straightforwardly extended to sampling transaction sequences from the database graphs.

\section{Problem Definition and Baseline}\label{sec:prob}

In this section, we define the problem formally and then present a baseline. We also show that the problem is indeed \#P-hard.

\subsection{Problem Definition} 

Let $I$ be a set of \emph{items}.  An \emph{itemset} $X$ is a subset of $I$, that is, $X \subseteq I$.  A \emph{transaction} is a tuple $T=(tid, X)$, where $tid$ is a unique transaction-id and $X$ an itemset. A transaction $T=(tid, X)$ is said to \emph{contain} itemset $Y$ if $Y \subseteq X$. In such a case, we overload the subset symbol and write $Y \subseteq T$. A \emph{transaction database} $\mathcal{T}$ is a set of transactions.

A \emph{database graph} is a directed graph, denoted by $G=(V,E, \mathbb{T})$, where $V$ is a set of vertices, $E =\{(u, v) \mid u, v \in V\} \subseteq V \times V$ is a set of directed edges and $\mathbb{T}$ is a set of transaction databases. 
Each vertex $v \in V$ is associated with a transaction database $\mathcal{T}_v\in \mathbb{T}$. 
%Without loss of generality, we assume that the transactions in in the transaction databases are unique across the whole graph. 
For the sake of simplicity, hereafter a database graph is also called a \emph{graph}.

A \emph{(directed) path} $p$ in a directed graph $G$ is a sequence of vertices $p=\langle v_1, \ldots, v_h \rangle$, where $(v_i, v_{i+1}) \in E$ ($1 \leq i < h)$, and the \emph{length} of the path is $len(p)=h-1$, the number of edges in the path. 
In this paper, we consider simple paths only, that is, a vertex appears in a path at most once. For two vertices $v_i,v_j$ in $G$, the \emph{distance} from $v_i$ to $v_j$, denoted by $dist(v_i, v_j)$, is the length of the shortest path from $v_i$ to $v_j$ in $G$. For a vertex $v \in V$, the \emph{$l$-neighborhood} of $v$ is $N_l(v)=\{v_i \mid dist(v, v_i) \leq l\}$. We define the \emph{$l$-th order degree} of $v$ as the number of vertices of distance $l$ from $v$, that is,
\begin{equation}
\label{eq:degree}
	d_l(v)=
	\begin{cases}
		1  & \text{if } l=0 \\
		\|\{v_i \mid dist(v, v_i)=l\}\|  & \text{if }l>0
	\end{cases}
\end{equation}
Apparently,  for $l > 0$, $d_l(v)=\| N_l(v) - N_{l-1}(v)\|$.

We are interested in sequential patterns carried by paths in a graph.  
Formally, a \emph{transaction sequence} in $G$, denoted by $ts = \langle T_1, \ldots, T_h \rangle$, is a sequence of transactions such that there exists a path $p=\langle v_1, \ldots, v_h \rangle$ and $T_i \in \mathcal{T}_{v_i}$ for $1 \leq i \leq h$.  
In such the case, we say $ts$ is \emph{supported} by $p$ and the length of $ts$ is $len(ts)=h-1$.

A \emph{sequential pattern} $s=\langle X_1, \ldots, X_h \rangle$ is a series of itemsets and the length of the sequential pattern is $len(s)=h-1$. 
A transaction sequence $ts = \langle T_1, \ldots, T_h \rangle$ \emph{contains} a sequential pattern $s=\langle X_1, \ldots, X_h \rangle$ if $X_i \subseteq T_i$ for $1 \leq i \leq h$. 
A \emph{sequential pattern} is also called a \emph{pattern}.

Denote $D_l$ as the set of all length-$l$ transaction sequences supported by paths in $G$, 
the \emph{frequency} of pattern $s$ in $D_l$, denoted by $f(s)$, is the proportion of unique transaction sequences in $D_l$ that contain $s$.
%Example~\ref{ex:concepts} gives some examples of the previously defined concepts.

\begin{example}[Concepts]\label{ex:concepts}
Figure~\ref{fig:concepts} shows a graph $G$, where each vertex is associated with a transaction database. 
Consider vertex $v_1$.  The \emph{$1$-}, \emph{$2$-} and \emph{$3$-neighborhoods} of $v_1$ are $N_1(v_1)=\{v_2, v_3\}$, $N_2(v_1)=\{v_2, v_3, v_4, v_5, v_6\}$ and $N_3(v_1)=\{v_2, v_3, v_4, v_5, v_6, v_7\}$, respectively. The first, second, and third degrees of $v_1$ are $d_1(v_1)=2$, $d_2(v_1)=3$ and $d_3(v_1)=1$, respectively. Starting from $v_1$, there are 6 paths, $p_1=\langle v_1, v_2 \rangle$, $p_2=\langle v_1, v_3 \rangle$, $p_3=\langle v_1, v_2, v_4 \rangle$, $p_4=\langle v_1, v_2, v_5 \rangle$, $p_5=\langle v_1, v_3, v_6 \rangle$ and $p_6=\langle v_1, v_2, v_4, v_7 \rangle$. 
Since $p_6$ is the only length-$3$ path in $G$, the length-$3$ transaction sequence set $D_3$ in $G$ consists of all the transaction sequences supported by $p_6$.
Figure~\ref{tab:exPatterns} shows the transaction sequences in $D_3$. 
\end{example}

\begin{figure}[t]
\centering
\includegraphics[width=86mm]{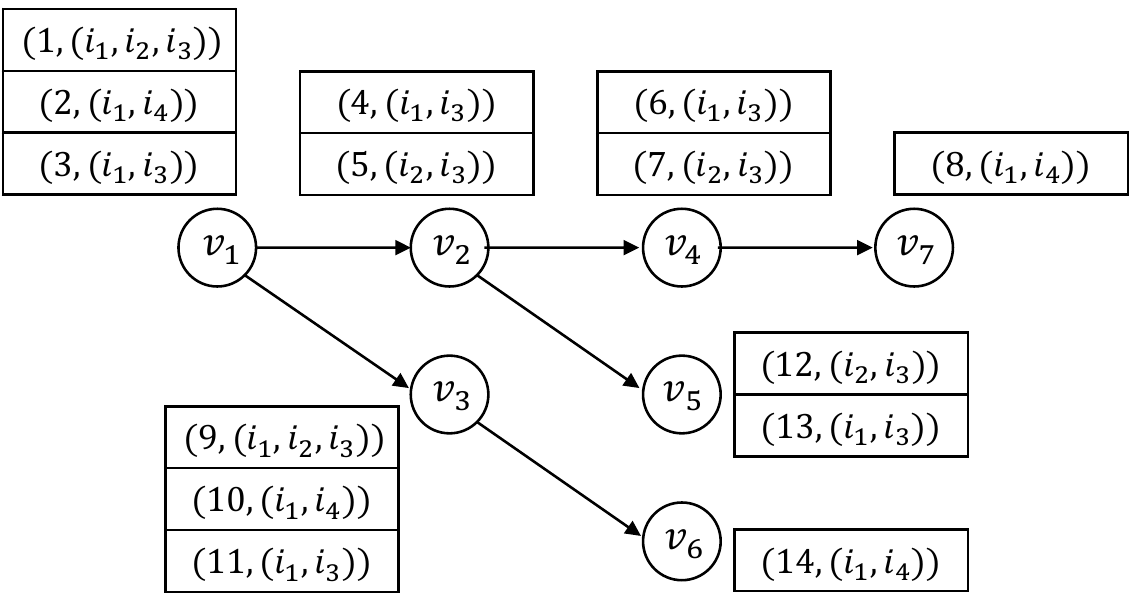}
\caption{An example of a graph $G$, where each vertex is associated with a transaction database.}
\label{fig:concepts}
\end{figure}

\begin{figure}[t]
\centering
	\begin{tabular}{|c|}
	\hline
	%The Set of Transaction Sequences \\ \hline
	$\langle (i_1, i_2, i_3)(i_1, i_3)(i_1, i_3)(i_1, i_4) \rangle$ \\ 
	$\langle (i_1, i_2, i_3)(i_1, i_3)(i_2, i_3)(i_1, i_4) \rangle$ \\
	$\langle (i_1, i_2, i_3)(i_2, i_3)(i_1, i_3)(i_1, i_4) \rangle$ \\ 
	$\langle (i_1, i_2, i_3)(i_2, i_3)(i_2, i_3)(i_1, i_4) \rangle$ \\
	$\langle (i_1, i_4)(i_1, i_3)(i_1, i_3)(i_1, i_4) \rangle$ \\
	$\langle (i_1, i_4)(i_1, i_3)(i_2, i_3)(i_1, i_4) \rangle$ \\
	
	$\langle (i_1, i_4)(i_2, i_3)(i_1, i_3)(i_1, i_4) \rangle$ \\
	$\langle (i_1, i_4)(i_2, i_3)(i_2, i_3)(i_1, i_4) \rangle$ \\
	$\langle (i_1, i_3)(i_1, i_3)(i_1, i_3)(i_1, i_4) \rangle$ \\
	$\langle (i_1, i_3)(i_1, i_3)(i_2, i_3)(i_1, i_4) \rangle$ \\
	$\langle (i_1, i_3)(i_2, i_3)(i_1, i_3)(i_1, i_4) \rangle$ \\
	$\langle (i_1, i_3)(i_2, i_3)(i_2, i_3)(i_1, i_4) \rangle$ \\

	\hline
	\end{tabular}
\caption{The set of length-$3$ transaction sequences $D_3$ of the graph $G$ in Figure~\ref{fig:concepts}.\label{tab:exPatterns}}
\end{figure}

\begin{figure}[t]
\newcommand{\tabincell}[2]{\begin{tabular}{@{}#1@{}}#2\end{tabular}}
\centering
	\begin{tabular}{|c|}
	\hline
	%Frequent Patterns \\ \hline
	$\langle (i_1)(i_3)(i_3)(i_1) \rangle$\\ 
	$\langle (i_1)(i_3)(i_3)(i_4) \rangle$\\
	$\langle (i_1)(i_3)(i_3)(i_1, i_4) \rangle$\\ 
	$\langle (i_3)(i_3)(i_3)(i_1) \rangle$\\
	$\langle (i_3)(i_3)(i_3)(i_4) \rangle$\\
	$\langle (i_3)(i_3)(i_3)(i_1, i_4) \rangle$\\
	$\langle (i_1, i_3)(i_3)(i_3)(i_1) \rangle$\\
	$\langle (i_1, i_3)(i_3)(i_3)(i_4) \rangle$ \\
	$\langle (i_1, i_3)(i_3)(i_3)(i_1, i_4) \rangle$\\
	\hline
	\end{tabular}
\caption{The top-$9$ length-$3$ frequent sequential patterns in $D_3$ of the graph $G$ in Figure~\ref{fig:concepts}. \label{tab:exPatterns2}}
\end{figure}

Now, we are ready to define the \emph{Top-$k$ Sequential Pattern Mining  problem in Graph (TSPMG) as follows. } 

\begin{problem}[TSPMG]
\label{the_problem}
Given a graph $G$, an integer $k > 0$ and a fixed length $l>0$, the problem of finding the top-$k$ sequential patterns is to find the top-$k$ patterns of length $l$ that have the largest frequencies in $D_l$ of $G$.
\end{problem}

\begin{example}[TSPMG]
\label{ex2}
Consider the graph $G$ in Figure~\ref{fig:concepts}. Let $k=9$ and $l=3$. The top-$9$ length-$3$ patterns, which have the largest frequencies in $D_3$ (shown in Figure~\ref{tab:exPatterns}), are 
\nop{In $G$, there is only one length-$3$ path $p_6=\langle v_1, v_2, v_4, v_7 \rangle$, which induces the length-$3$ transaction sequence set $D_3$ . }
listed in Figure~\ref{tab:exPatterns2}.

\nop{Carried by $p_6$, we can extract 12 transaction sequences as shown in Table~\ref{tab:exPatterns}. In the same table, we also show the top-9 sequential patterns sorted by their frequency.}
\end{example} 

Table~\ref{tab:notations} summarizes the frequently used notations.

\begin{table}[t]

\caption{Frequently used notations.}
\centering
\label{tab:notations}
	\begin{tabular}{|p{1.8cm}<{\centering}|p{6cm}|}
	\hline
	Notation & \multicolumn{1}{|c|}{Description}\\
	\hline
%	$G=(V,E)$ & A directed graph. $V$ is the vertex set and $E$ is the edge set.\\
%	\hline
%	$|\cdot|$ & The volume operator. \\
%	\hline
%	$\langle\cdot\rangle$ & The notation of an ordered sequence. \\
%	\hline
	$dist(v_i,v_j)$ & The distance from $v_i$ to $v_j$. \\ 
	\hline
	$d_i(v)$ & The $i$-th order degree of $v$. \\ 
	\hline
	$p$ & A \emph{(directed) path} in $G$, which is a sequence of vertices. \\ 
	\hline
	
	\nop{
	$ts$  & A transaction sequence, which is a sequence of transactions.  \\ 
	\hline
	$s$ & A sequence, also named as pattern. \\ 
	\hline}
	
	\nop{
	$\sigma(s)$ & The support of pattern $s$.\\
	\hline}
	
	$f(s)$ & The real frequency of pattern $s$. \\
	\hline
	$\hat{f}(s)$ & The unbiased pattern frequency estimator.\\
	\hline
	$D_l$ & The set of all length-$l$ transaction sequences in $G$.\\
	\hline
	$P_l$ & The set of all length-$l$ paths in $G$. \\
	\hline
	$S_l$ & The set of sampled length-$l$ transaction sequences. \\
	\hline
	$I$ & The set of all items in $G$. \\
	\hline
\end{tabular}
\end{table}

% !TEX root = ../Sequential Patterns in Database Graphs.tex

\subsection{Hardness of the TSPMG Problem}

In this section, we analyze the hardness of the TSPMG problem.
We define the \emph{Sequential Pattern Counting problem in Graph} (SPCG) and prove that the SPCG problem is \#P-hard.
In sequel, we prove that the TSPMG problem is at least as hard as the SPCG problem.

\begin{problem}[SPCG]
\label{count_problem}
Given a database graph $G$, a threshold $\alpha > 0$ and a fixed length $l>0$, the problem of counting sequential patterns is to count the number of length-$l$ patterns with frequencies larger than $\alpha$ in $D_l$ of $G$.
\end{problem}

\begin{theorem}
\label{th:hardness_1}
The SPCG problem is \#P-hard.
\end{theorem}
\begin{proof}[sketch]

Given a sequence database $SDB$, for each sequence in the database that contains a list of transactions, we can construct a graph that contains a single path such that a node is set up for each transaction and the nodes are linked into a path according to the order of transactions in the sequence.  Moreover, for a single path $T_1 \rightarrow T_2 \rightarrow \cdots \rightarrow \nop{T_l}{T_{l+1}}$, we add edges $(T_i, T_j)$ for $i < j$.  $\frac{l (l-1)}2$ edges are added on a path of length $l$. In this way, a sequence in $SDB$ is transformed into a directed acyclic graph (DAC).

Then, we combine all the DACs into a database graph $G$.  That is, the graph is a collection of DACs, one representing a sequence in $SDB$. Thus, each node contains a transaction database where there is only one transaction.  Apparently, this reduction step takes polynomial time.

It can be shown easily that a subsequence is a sequential pattern in $SDB$ if and only if it is a sequential pattern in the database graph constructed.
\end{proof}

\begin{theorem}
\label{th:hardness_2}
The TSPMG problem is at least as hard as the SPCG problem.
\end{theorem}
\begin{proof}
We prove this by a Cook reduction from the SPCG problem.
Denote by $\Theta(D_l, k)$ the set of top-$k$ length-$l$ patterns in $G$. Let $s_k=\Theta(D_l, k)\setminus \Theta(D_l, k-1)$ be the pattern with the $k$-th largest frequency $f(s_k)$ in $D_l$ of $G$.

Given the oracle of the TSPMG problem, we can search for the parameter $k$ that satisfies $f(s_{k}) > \alpha \geq f(s_{k+1})$ by querying the oracle multiple times. 
Such $k$ is exactly the answer to the SPCG problem.

Let $I$ be the set of all items in $G$. We have $k\in[1, 2^{|I|(l+1)}]$. Thus, the time of a binary search for $k$ is in $\mathcal{O}(|I|(l+1))$.
For each $k$, we can obtain $s_{k}$ and $s_{k+1}$ by querying the oracle to get $\Theta(D_l, k-1)$, $\Theta(D_l, k)$ and $\Theta(D_l, k+1)$; $f(s_k)$ and $f(s_{k+1})$ can be computed by linearly searching $D_l$ in $\mathcal{O}(|V|^{(l+1)}C^{(l+1)})$ time.
Therefore, the overall time to solve the SPCG problem is in $\mathcal{O}(|I|(l+1)  |V|^{(l+1)}C^{(l+1)})$. 
Recall that $l$ is the constant path length, we know that the SPCG problem is Cook reducible to the TSPMG problem. %The theorem follows.
\end{proof}

\subsection{Baseline}
\label{sec:exactAl}
Given a graph $G$, how can we find interesting patterns on directed paths? 
Intuitively, with a given path length $l$, we can enumerate all paths in $G$ and then extract all transaction sequences supported by the paths. Then, by applying traditional sequential pattern mining techniques~\citep{Pei2001,Agrawal1995}, we can find the exact top-$k$ sequential patterns among transaction sequences. Algorithm~\ref{al:groundTruth} gives the details of this method. 

\begin{algorithm}[t]
\caption{The Baseline Method}
\label{al:groundTruth}
\KwIn{A database graph $G$, path length $l$ and $k$.}
\KwOut{The set of top-$k$ length-$l$ sequential patterns $\Theta(D_l, k)$.}
\BlankLine
\begin{algorithmic}[1]
    \STATE Initialize length-$l$ transaction sequence set $D_l\leftarrow\emptyset$.
    \FOR{each length-$l$ path $p$ in $G$}
    	\FOR{each transaction sequence $ts$ induced by $p$}
		\STATE $D_l \leftarrow D_l\cup ts$.
	\ENDFOR
    \ENDFOR
    \STATE Apply sequential pattern mining techniques on $D_l$ to obtain $\Theta(D_l, k)$.
\BlankLine
\RETURN $\Theta(D_l, k)$.
\end{algorithmic}
\end{algorithm}

%The brute force enumeration in Algorithm~\ref{al:groundTruth} produces the exact result of top-$k$ sequential patterns. 
%However, 
The volume of $D_l$, denoted by $|D_l|$, is in $\mathcal{O}(|V|^{(l+1)}C^{(l+1)})$, where $|V|$ is the number of vertices in $G$ and $C=\max_{v_i\in V} |\mathcal{T}_{v_i}|$ is the maximum number of transactions in the transaction database of a single vertex of $G$. 
As a result, the baseline method is computationally expensive due to the large amount of time and space needed to enumerate and store the transaction sequences in $D_l$. 

In the next section, we tackle this problem by an efficient sampling method that significantly reduces the amount of transaction sequences needs to search, and thus achieves high-quality approximation results.

\section{A Fast Sampling-based Method}
\label{sec:sampleAl}

The key idea of our fast sequential pattern mining method is to first estimate the pattern frequencies using a set of transaction sequences sampled from $G$, then obtain the top-$k$ sequential patterns using the estimated pattern frequencies. In this section, we first introduce our \emph{two-step sampling framework} that randomly samples transaction sequences. Then, we provide an unbiased estimator to estimate pattern frequency using the sampled transaction sequences. Last, we prove the upper bound of the complexity for the sampling-based algorithm.

\subsection{A Two-step Sampling Framework}
The two-step sampling framework involves a path sampling method that uniformly samples a set of length-$l$ paths from $G$ and a transaction sequence sampling method that uniformly samples transaction sequences from each sampled path. 

Notice that, although the first step sampling perform uniform sampling of the paths on the graph and the second step sampling perform uniform sampling of transaction sequence on the sampled paths,  respectively, the overall two-step sampling is not uniform with respect to the length-$l$ transaction sequences, since the numbers of transactions are different in different vertices.
As a result, we cannot directly mine frequent sequential patterns from the sampled transaction sequences by simply using the support of patterns in the sampled transaction sequences.
However, as illustrated in Section~\ref{sec:upfe}, we can still mine the top-$k$ sequential patterns by estimating the pattern frequencies with bounded estimation error using a carefully designed unbiased estimator.

Denote by $P_l$ the set of all length-$l$ paths in $G$. The \emph{path sampling method} samples a length-$l$ ($l \geq 1$) path $p=\langle v_1,\ldots, v_q, \ldots, v_{l+1} \rangle$ from $P_l$ by progressively sampling an ordered sequence of $(l+1)$ vertices, where each vertex $v_q$ in $p$ is sampled with probability
\begin{equation}
\label{eq:pathSampling}
	\mathbb{P}(v_q)=
	\begin{cases}
	\mathbb{P}(v_1)=\frac{d_l(v_1)}{\sum_{v_j \in V}d_l(v_j)}  & q=1\\
	\mathbb{P}(v_q|v_{q-1})=\frac{d_{l-q+1}(v_q)}{d_{l-q+2}(v_{q-1})}  & 2\leq q \leq (l+1)
	\end{cases}
\end{equation}
  
\begin{algorithm}[t]
\caption{Path Sampling Method}
\label{al:pathSampling}
\KwIn{Graph $G$ and path length $l$.}
\KwOut{A length-$l$ path $p$.}
\BlankLine
\begin{algorithmic}[1]
\STATE $p\leftarrow\emptyset$.
         \FOR{$q=1 \to l$}
                	\STATE Sample vertex $v_q$ with the probability given by Equation~\ref{eq:pathSampling}.
                	\STATE $p\leftarrow p \cup v_q$.
         \ENDFOR
\BlankLine
\RETURN $p$.
\end{algorithmic}
\end{algorithm}

Algorithm~\ref{al:pathSampling} gives the pseudocode of the path sampling step. We prove that Algorithm~\ref{al:pathSampling} conducts a uniform sampling on paths.

\begin{theorem}
\label{th:pathSampling}
Given a graph $G$ and a path length $l$, Algorithm~\ref{al:pathSampling} uniformly samples a path $p$ from $P_l$.
\end{theorem}

\begin{proof}
Since $P_l$ is the set of all length-$l$ paths in $G$, to show that the path sampling in Algorithm~\ref{al:pathSampling} is uniform, we only need to prove that the probability of sampling a length-$l$ path $p$ is $\frac{1}{|P_l|}$.

Denote by $\mathbb{P}(p)$ the probability of sampling a path $p=\langle v_1,\ldots, v_{l+1} \rangle$ in $G$. According to Algorithm~\ref{al:pathSampling} and Equation~\ref{eq:pathSampling}, we have 
\begin{equation}
\label{eq:pathsample_proof}
\begin{aligned}
	\mathbb{P}(p)=& \mathbb{P}(v_1) \times \mathbb{P}(v_2) \times \cdots \times \mathbb{P}(v_{l+1}) \\
	=& \mathbb{P}(v_1) \times \mathbb{P}(v_2|v_1) \times \cdots \times \mathbb{P}(v_{l+1}|v_l)\\
	\nop{=& \frac{d_l(v_1)}{\sum_{v_j \in V}d_l(v_j)} \times \frac{d_{l-1}(v_2)}{d_{l}(v_1)} \times \cdots \times \frac{d_{0}(v_{l+1})}{d_{1}(v_{l})}\\}
	=& \frac{d_l(v_1)}{\sum_{v_j \in V}d_l(v_j)} \times \frac{d_{l-1}(v_2)}{d_{l}(v_1)} \times \cdots \times \frac{1}{d_{1}(v_{l})}\\
	=& \frac{1}{\sum_{v_j \in V}d_l(v_j)}\\
	=& \frac{1}{|P_l|}
\end{aligned}
\end{equation}
The theorem holds.
\end{proof}
 
Algorithm~\ref{al:pathSampling} uniformly samples one path from all length-$l$ paths in $G$. To sample a set of length-$l$ paths with replacement, we simply run Algorithm~\ref{al:pathSampling} multiple times. 

Next, we introduce the transaction sequence sampling method.
Denote by $p^i=\langle v^i_1,\ldots,$ $ v^i_q,\ldots,v^i_{l+1} \rangle$ the $i$-th length-$l$ path sampled by Algorithm~\ref{al:pathSampling}, the \emph{transaction sequence sampling method} in Algorithm~\ref{al:seqSampling} uniformly samples a transaction sequence $ts^i$ with probability $\frac{1}{\prod_{q=1}^{l+1}|\mathcal{T}_{{v_q}^{i}}|}$ from all  transaction sequences induced by $p^i$.

\begin{algorithm}[t]
\caption{Transaction Sequence Sampling Method}
\label{al:seqSampling}
\KwIn{The $i$-th sampled length-$l$ path $p^i=\langle v^i_1,\ldots,$ $ v^i_q,\ldots,v^i_{l+1} \rangle$.}
\KwOut{A sampled transaction sequence $ts^i$.}
\begin{algorithmic}[1]
\STATE $ts^i \gets \emptyset$.
        	\FOR{$q=1 \to (l+1)$}
        	    	\STATE Sample transaction $T^i_q$ from vertex $v_q^i$ with the probability $\frac{1}{|\mathcal{T}_{{v_q^i}}|}$.
	    	\STATE $ts^i\gets ts^i \cup T^i_q$.
        	\ENDFOR
\BlankLine 
\RETURN $ts^i$.
\end{algorithmic}
\end{algorithm}

\begin{algorithm}[t]
\caption{Two-step Sampling Framework}
\label{al:sampling}
\KwIn{Graph $G$, path length $l$ and sample size $m$.}
\KwOut{A set $S_l$ of uniformly sampled length-$l$ transaction sequences.}
\begin{algorithmic}[1]
	\FOR{$each \ v_j \in V$}
            	\STATE Calculate degrees $\{d_q(v_j)|1 \leq q \leq l\}$ using Equation~\ref{eq:degree}.
        	\ENDFOR       
	\STATE $\psi\gets \emptyset$, $S_l\gets \emptyset$
        	\FOR{$i=1 \to m$}
            	\STATE Sample a path $p^{i}=\langle v_1^{i}, \ldots, v_{l+1}^{i} \rangle$ from $G$ using Algorithm~\ref{al:pathSampling}.
            	\STATE $\psi\gets \psi\cup p^i$.
        \ENDFOR
        
        	\FOR{each $p^i$ in $\psi$}
            	\STATE Sample a transaction sequence $ts^{i}$ from $p^i$ using Algorithm~\ref{al:seqSampling}.
            	\STATE $S_l\gets S_l\cup ts^i$.
        \ENDFOR
 \BlankLine
 \RETURN $S_l$.
 \end{algorithmic}
\end{algorithm}

The \emph{two-step sampling framework} is summarized in Algorithm~\ref{al:sampling}, which samples a length-$l$ transaction sequence set $S_l$ from $G$ in two steps.
First, we uniformly sample $m$ length-$l$ paths from $G$ using Algorithm~\ref{al:pathSampling}. Then, for each sampled path $p^i$, we uniformly sample one transaction sequence $ts^i$ using Algorithm~\ref{al:seqSampling}. 
%In the following, we introduce how to make unbiased estimation of pattern frequency by using $S_l$ and give the lower bound of the sample size $|S_l|$.

\subsection{Unbiased Pattern Frequency Estimator}
\label{sec:upfe}
In this section, we first define some useful notations. Then, we introduce an unbiased estimator of pattern frequency.

For a length-$l$ path $p^i=\langle v_1^{i}, \ldots, v_q^i, \ldots, v_{l+1}^{i} \rangle$, we assign a unique index to each transaction sequence supported by $p^i$, and denote by $ts^i_j$ the $j$-th transaction sequence supported by $p^i$.
Apparently, $M^i=\prod^{l+1}_{i=1}|\mathcal{T}_{v_q^i}|$ is the number of all transaction sequences supported by $p^i$.

We further define a random variable $Y^i_j(h)$  
\begin{equation}
\label{eq:Y}
	Y^i_j(h)=
	\begin{cases}
	1 & \text{{if the $h$-th sample draws }}ts^i_j\\
	0  & \text{otherwise}
	\end{cases}
\end{equation}
where $1 \leq h \leq |S_l|$ is the index of a sample and $|S_l|$ is the volume of $S_l$. Since $ts^i_j$ is drawn with probability $\frac{1}{|P_l| M^i}$, we have $\mathbb{E}(Y^i_j(h))=\frac{1}{|P_l| M^i}$.

Now we present an unbiased pattern frequency estimator.
Denote by $P_l$ the set of all length-$l$ paths in graph $G$. For a pattern $s= \langle X_1\ldots, X_q, \ldots, X_{l+1} \rangle$, an estimator of the real pattern frequency $f(s)$ is
\begin{equation}
\label{eqn:fhat}
\hat{f}(s)=
\frac{\sum_{i=1}^{|P_l|}\sum_{j=1}^{M^i} \sum_{h=1}^{|S_l|} Y^i_j(h) W^i_j}
{\frac{|S_l|}{|P_l|}
\sum_{i=1}^{|P_l|}M^i
}
\end{equation}
where $W^i_j$ is defined as
\begin{equation}
\label{eq:IncW}
	W^i_j=
	\begin{cases}
	M^i & s \subseteq ts^i_j\\
	0  & \text{otherwise}
	\end{cases}
\end{equation}
Note that, since $\mathbb{P}(s\subseteq ts^i_j)=f(s)$, we have $\mathbb{E}(W^i_j)=f(s)M^i$.

Now, we prove that $\hat{f}(s)$ is an unbiased estimator of $f(s)$.

\begin{theorem}
For a pattern $s$, $\hat{f}(s)$ is an unbiased estimator of $f(s)$.

\label{th:patternEs2}
\end{theorem}

\begin{proof}
First, since $Y^i_j(h)$ and $W^i_j$ are independent random variables, we have 
\begin{equation}\nonumber
\mathbb{E}(Y^i_j(h) W^i_j)
=\mathbb{E}(Y^i_j(h)) \mathbb{E}(W^i_j)
=\frac{f(s)}{|P_l|}
\end{equation}
Then, we have
\begin{equation}\nonumber
\begin{aligned}
	\mathbb{E}[\hat{f}(s)]
	=& 
\frac
{
\sum_{i=1}^{|P_l|} \sum_{j=1}^{M^i} \sum_{h=1}^{|S_l|} \mathbb{E}(Y^i_j(h) W^i_j)
}
{
\frac{|S_l|}{|P_l|}
\sum_{i=1}^{|P_l|}M^i
}	
\\
	=& 
\frac
{
\sum_{i=1}^{|P_l|} \sum_{j=1}^{M^i} 
\frac{|S_l|}{|P_l|}f(s)
}
{
\frac{|S_l|}{|P_l|}
\sum_{i=1}^{|P_l|}M^i
}	
\\
	=& f(s)
\end{aligned}
\end{equation}
The proof follows.
\end{proof}

Using the estimator $\hat{f}(s)$ (Equation~\ref{eqn:fhat}) to calculate the top-$k$ most frequent patterns is efficient. In our problem, we only care about the order of the pattern frequencies. As we have shown above, the the denominator of $\hat{f}(s)$, that is, ${\frac{|S_l|}{|P_l|}\sum_{i=1}^{|P_l|}M^i}$, won't affect the order of pattern frequencies. Therefore, the patterns can be sorted by the nominator of $\hat{f}(s)$, which is the weighted sum of the count of transactions in $S_l$ that contains pattern $s$. For example, for a pattern $s$, if it is contained in $ts^1$ and $ts^2$, the weighted sum is equal to $M^1+M^2.$

\nop{Using the estimator $\hat{f}(s)$ (Equation~\ref{eqn:fhat}) to calculate the top-$k$ most frequent patterns is efficient. 
Since we only care about the order of pattern frequencies, we can ignore the denominator of $\hat{f}(s)$, that is, ${\frac{|S_l|}{|P_l|}
\sum_{i=1}^{|P_l|}M^i
}$, which is constant to all patterns. The nominator is simply the weighted sum of the count of transactions in $S_l$ that contains pattern $s$.}

\subsection{Bounding the Sample Size}
In this section, we analyze the approximation error of $\hat{f}(s)$ in Theorem~\ref{thm:hoeffding} and derive \nop{a maximal lower bound of sample size}an upper bound of sample complexity for the proposed sampling-based algorithm in Theorem~\ref{eq:sampleSize} and Theorem~\ref{eq:lemma}. \nop{\todo{That is, we just need to obtain a \emph{maximum} sample size, which is not much bigger than a given error, with high probability.}}

\begin{theorem}[Chernoff's inequality \citep{Chernoff1952}]
\label{th:chernoff}
Let $Z_1,\ldots,Z_n$ be \emph{independent} random variables. They need not have the same distribution. Assume that $0 \leq Z_i \leq 1$ always, for each $i$. Let $Z=Z_1+\ldots+Z_n$. Write $\mu=\mathbb{E}[Z]=\mathbb{E}[Z_1]+\ldots+\mathbb{E}[Z_n]$. Then for any $\varepsilon \geq 0$,
\begin{equation}\nonumber
\mathbb{P}[|Z-\mu| \geq \varepsilon\mu] \leq 2\exp(-\frac{\varepsilon^2}{3}\mu)
\end{equation}
\end{theorem}

\begin{theorem}
\label{thm:hoeffding}
Given a set of length-$l$ sampled transaction sequences $S_l$ and an arbitrary pattern $s$, for a fixed threshold $\varepsilon\geq 0$, we have:
\begin{equation}\nonumber
\mathbb{P}(|\hat{f}(s) - f(s)|\geq \varepsilon) \leq    2\exp(-\frac{\varepsilon^2 a}{3}|S_l|)
\end{equation}
where $a=\frac{1}{|P_l|M^*} \sum_{i=1}^{|P_l|}M^i$ and $M^*=\max_i M^i$.
\end{theorem}
\begin{proof}
We define a random variable $U^i_j(h)=Y^i_j(h)W^i_j$.
According to Equation~\ref{eq:Y} and Equation~\ref{eq:IncW}, we have $U^i_j(h)\in [0, M^i]$. Thus, $\frac{U^i_j(h)}{M^*}\in [0, 1]$.

Now, consider random variable $$U=\sum_{i=1}^{|P_l|} \sum_{j=1}^{M^i} \sum_{h=1}^{|S_l|} U^i_j(h).$$ Applying Theorem~\ref{th:chernoff}, it follows
\begin{equation}\nonumber
\mathbb{P}(|\frac{U}{M^*} - \frac{\mathbb{E}(U)}{M^*}| \geq t \frac{\mathbb{E}(U)}{M^*}) \leq 2 \exp(-\frac{t^2}{3} \frac{\mathbb{E}(U)}{M^*})
\end{equation}

Since $\mathbb{E}(U) = f(s) \frac{|S_l|}{|P_l|} \sum_{i=1}^{|P_l|}M^i $ and $\frac{U}{\frac{|S_l|}{|P_l|} \sum_{i=1}^{|P_l|}M^i}=\hat{f}(s)$ (see Equation~\ref{eqn:fhat}), it follows
\begin{equation}\nonumber
\mathbb{P}(|\hat{f}(s) - f(s)| \geq t f(s)) \leq 2 \exp(- \frac{t^2 f(s)}{3} \frac{|S_l|}{|P_l|M^*} \sum_{i=1}^{|P_l|}M^i)
\end{equation}

Since $f(s) \in [0,1]$, it follows
\begin{equation}\nonumber
\mathbb{P}(|\hat{f}(s) - f(s)| \geq t f(s)) \leq 2 \exp(- \frac{t^2 f^2(s)}{3} \frac{|S_l|}{|P_l|M^*} \sum_{i=1}^{|P_l|}M^i)
\end{equation}

Let $\varepsilon= t f(s)$ and $a=\frac{1}{|P_l|M^*} \sum_{i=1}^{|P_l|}M^i$. It follows
\begin{equation}\nonumber
\mathbb{P}(|\hat{f}(s) - f(s)|\geq \varepsilon) \leq    2\exp(-\frac{\varepsilon^2 a}{3}|S_l|)
\end{equation}
The proof follows.
\end{proof}

In Theorem~\ref{eq:sampleSize}, we further derive \nop{a maximal lower bound of sample size}an upper bound of sample complexity for the proposed sampling-based algorithm to make sure for each pattern $s$, $|\hat{f}(s) - f(s)| \leq \varepsilon$ holds with a high probability. 

\begin{theorem}
\label{eq:sampleSize}
Denote by $Q_l$ the set of all length-$l$ patterns in $G$. If $|S_l|\geq \frac{12}{\varepsilon^2 a} ln\frac{2|Q_l|}{\delta}$, then $\mathbb{P}(\forall s\in Q_l, |\hat{f}(s) - f(s)| < \frac{\varepsilon}{2}) \geq 1 - \delta$.
\nop{
For fixed $\varepsilon,\delta \in (0,1)$, by setting the number of sample size $|S| \geq \frac{12}{\varepsilon^{2}}ln\frac{2|Q|}{\delta}$, with 
probability at least $1-\delta$, the following condition holds. (1)For each $s$, $|\frac{f_S(s)}{|S|}-\frac{f_D(s)}{|D|}| \leq \frac{\varepsilon}{2}$ ($Q$ is the set of all sequential patterns.)}
\end{theorem}
\begin{proof}
Since $Q_l$ is the set of all length-$l$ patterns in $G$, by Theorem~\ref{thm:hoeffding} and the Union Bound \citep{Bonferroni1936}, it follows
\begin{equation}\nonumber
\mathbb{P}(\exists s\in Q_l, |\hat{f}(s) - f(s)| \geq \frac{\varepsilon}{2}) \leq 2|Q_l|\exp(-\frac{\varepsilon^2a}{12} |S_l|)
\end{equation}

Thus, we have
\nop{
\begin{equation}\nonumber
\mathbb{P}(\nexists s\in Q, |\hat{f}(s) - f(s)| \geq \frac{\varepsilon}{2}) \geq 1 - 2|Q|\exp(-\frac{\varepsilon^2a}{12} |S|)
\end{equation}

This further indicates that
}
\begin{equation}\nonumber
\mathbb{P}(\forall s\in Q_l, |\hat{f}(s) - f(s)| < \frac{\varepsilon}{2}) \geq 1 - 2|Q_l|\exp(-\frac{\varepsilon^2a}{12} |S_l|)
\end{equation}

Since $|S_l|\geq \frac{12}{\varepsilon^2 a} ln\frac{2|Q_l|}{\delta}$, it follows that
\begin{equation}\nonumber
1 - 2|Q_l|\exp(-\frac{\varepsilon^2a}{12} |S_l|) \geq 1-\delta
\end{equation}

Thus, we have
\begin{equation}\nonumber
\mathbb{P}(\forall s\in Q_l, |\hat{f}(s) - f(s)| < \frac{\varepsilon}{2}) \geq 1 - \delta
\end{equation}
The proof follows.
\end{proof}

The bound of $\frac{12}{\varepsilon^2 a} ln\frac{2|Q_l|}{\delta}$ in Theorem~\ref{eq:sampleSize} is an upper bound of sample complexity for the proposed sampling-based algorithm, since the Union Bound \citep{Bonferroni1936} is applied to bound the approximation errors of pattern frequencies of all patterns in $Q_l$.
As demonstrated later in the experiments, in practice we can achieve a good estimation quality with a sample size that is smaller than the bound in Theorem~\ref{eq:sampleSize}. Next, we analyze the relationship between $|S_l|$ and path length $l$ in Theorem~\ref{eq:lemma}.

\begin{theorem}
\label{eq:lemma}
Denote by $I$ the set of all items in $G$, and by $Q_l$ the set of all length-$l$ patterns in $G$.
If $|S_l|\geq \frac{12|I|(l+1)+12}{\varepsilon^2 a}  ln\frac{2}{\delta}$, then $\mathbb{P}(\forall s\in Q_l, |\hat{f}(s) - f(s)| < \frac{\varepsilon}{2}) \geq 1 - \delta$.
\end{theorem}
\begin{proof}
Since $Q_l$ is the set of all patterns in $G$, we have
\begin{equation}\nonumber
|Q_l|\leq 2^{|I|(l+1)}
\end{equation}

By substituting the above inequality into the condition $|S_l|\geq \frac{12}{\varepsilon^2 a} ln\frac{2|Q_l|}{\delta}$ of Theorem~\ref{eq:sampleSize}, the theorem follows.
\end{proof}

According to Theorem~\ref{eq:lemma}, when $|S_l|\geq \frac{12|I|(l+1)+12}{\varepsilon^2 a}  ln\frac{2}{\delta}$, the estimated pattern frequency $\hat{f}(s)$ for each pattern $s$ will have a high probability to be close to its real frequency $f(s)$. Therefore, we can use a sequential pattern mining method, such as PrefixSpan~\citep{Pei2001}, to extract  sequential patterns from the set of sampled transaction sequences. 
The ranking on the estimated pattern frequencies will be a good approximation of the ranking of real pattern frequencies. As a result, we can obtain high-quality approximation of top-$k$ sequential patterns using the pattern frequencies estimated from the set of sampled transaction sequences.

Comparing to the baseline method that enumerates the exponential number of transaction sequences in $D_l$, the proposed sampling method significantly reduces the number of required transaction sequences and achieves guaranteed approximation quality.

\section{Experiments}
\label{sec:exp}
In this section, we analyze the effectiveness and efficiency of the proposed \emph{sampling-based method} and the \emph{baseline method}.
The algorithms were implemented in Java and compiled with JDK 8.
We use PrefixSpan~\citep{Pei2001} to extract frequent sequential patterns with the set of sampled transaction sequences. 
All experiments were conducted on a PC running Windows 7 operating system with Intel Core i7-3770 3.40 GHz CPU and 16GB main memory.

\subsection{Datasets}
\label{sec:datasets}
We use both synthetic and real datasets in our experiments.

\textbf{Synthetic datasets \mbl{(SYN1 and SYN2)}.} We use two synthetic datasets, denoted by SYN1 and SYN2, which are generated by the IBM Quest Synthetic Data Generator~\footnote{\url{http://www.almaden.ibm.com/cs/quest/syndata.html}}. 
Both SYN1 and SYN2 consist of the same graph structure with 28 edges and 24 vertices. The difference between SYN1 and SYN2 lies in the distribution of transaction databases.
In SYN1, each vertex has 20 transactions. 
In SYN2, each vertex with degree $d$ has $8(d+1)$ transactions in the associated transaction database.
%Thus, SYN2 is more general than SYN1.
For the path length $l=1$ and $l=2$, the numbers of all transaction sequences in SYN1 are 11,200 and 184,000, respectively, and the numbers of all transaction sequences in SYN2 are 8,828 and 141,312, respectively.
For both SYN1 and SYN2, the average number of items per transaction is 5. The total numbers of items in SYN1 and SYN2 are 96 and 91, respectively.

\textbf{Open flights dataset (Flight).} The open flights dataset is published in OpenFlights website~\footnote{https://openflights.org/data.html}.We built our graph by treating each airport as a vertex and each airline as a directed edge. Each transaction database of a vertex is an entry of routes, where each route contains the information of source airport, destination airport, codeshare, stops, equipment, and so on. The top-$k$ sequential patterns in the dataset show interesting flight services following the airlines.

\begin{table}[t]
\caption{The statistics of the datasets.}
\centering
\label{tab:subgraphs}
	\begin{tabular}{|c|c|c|c|c|}
	\hline
	%Data- & & & &Average \\
	%sets & \#Vertices & \#Edges & \#Trans & \#items/trans \\
	Datasets & \#Vertices & \#Edges & \#Trans & Average \#items/trans \\
	\hline
	\textbf{SYN1} & 24 & 28 & 480 & 5.00\\
	\hline
	\textbf{SYN2} & 24 & 28 & 420 & 5.00\\
	\hline
	\textbf{Flight} & 28 & 448 & 2,524 & 4.00\\
	\hline
	\textbf{CN1} & 16,198 & 60,856 & 120,741 & 6.65\\
	\hline
	\nop{
	\textbf{CN2} & 105,426 & 516,271 & 1,013,190 & 6.62\\
	\hline}
	\textbf{CN2} & 324,228 & 1,809,469 & 3,340,335 & 6.59\\
	\hline
	\end{tabular}
\end{table}

\textbf{Citation network datasets (CN1 and CN2).} The citation network datasets are built from the ArnetMiner dataset~\citep{Tang2008}. We built our graph datasets by treating each paper as a vertex and each citation as a directed edge. Each transaction database of a vertex is a set of topics, where each topic is represented by a set of keywords extracted from the abstracts of the paper. The top-$k$ sequential patterns in these datasets reveal interesting patterns of topic changes in paper citations and relations among research topics. 
By randomly sampling the original citation network of the ArnetMiner dataset, we built two citation network datasets, denoted by CN1 and CN2, that have different sizes. 

Some statistics of the datasets are listed in Table~\ref{tab:subgraphs}. Since the baseline method is too slow to run on large datasets, we build these small synthetic and real-world datasets to comprehensively assess and compare the efficiency and effectiveness of the sampling-based method with the baseline method.

\subsection{Quality Evaluation Metrics}
Denote by $\mathcal{G}$ the ground truth ranked list of the real top-$k$ patterns and by $\mathcal{L}$ the ranked list of patterns produced by the proposed sampling-based method. \nop{The length of $\mathcal{L}$, denoted by $|\mathcal{L}|$, is larger than $k$.}Let $\mathcal{G}(k)$ and $\mathcal{L}(k)$ be the ranked lists of the top-$k$ patterns in $\mathcal{G}$ and $\mathcal{L}$, respectively.
We use the following evaluation metrics.

\textbf{Mean Estimation Error (ME).} Mean estimation error is the average error between real pattern frequency $f(s)$ and estimated pattern frequency $\hat{f}(s)$ in Equation~\ref{eqn:fhat}. Denote $s_i$ as the $i$-th pattern in ranked list $\mathcal{G}(k)$, the ME at rank-$k$ is computed as
\begin{equation}
ME(k) = \frac{\sum_{i=1}^{k} |f(s_i) - \hat{f}(s_i)|}{k}
\end{equation}

\textbf{Average Precision (AP).} Average precision~\citep{AP2001} is widely used to evaluate the similarity between two ranked lists in the field of information retrieval.
Denote by $R(i)$ an indicator function that equals to 1 if the pattern at rank $i$ in $\mathcal{L}$ is contained in $\mathcal{G}$, the average precision of $\mathcal{L}$ when using $\mathcal{G}(k)$ as the ground truth is
\begin{equation}
AP(k) = \frac{\sum_{i=1}^{|\mathcal{L}|} \frac{|\mathcal{L}(i)\cap \mathcal{G}(k)|}{|\mathcal{L}(i)|} R(i)} {|\mathcal{G}(k)|}
\end{equation}

\textbf{Ranking Similarity (RS).} The ranking similarity~\citep{Kimura2006} quantifies the degree of similarity between two ranked lists at rank $k$. The ranking similarity of $\mathcal{G}$ and $\mathcal{L}$ at rank $k$ is
\begin{equation}
RS(k) = \frac{|\mathcal{G}(k) \cap \mathcal{L}(k)|}{k}
\end{equation}

%\textbf{Running Time (RT).} The running time measures the total time cost of each compared method. It is a widely adopted measurement for efficiency.

\subsection{Results on Synthetic Datasets}
%In this section, we analyze the AP, RS and RT performances of both the baseline method and the proposed sampling-based method on the synthetic datasets SYN1 and SYN2.

Since SYN1 and SYN2 are small, we are able to run the baseline method to obtain the exact pattern frequencies of all patterns and the exact ranked lists of the top-$k$ sequential patterns with $k=50$, $100$, $200$, $300$, and $1000$, respectively. 
The exact pattern frequencies are used to calculate the ME of the proposed sampling-based method. 
The exact ranked lists of patterns are used as the ground truth to evaluate the performance of the  sampling-based method in AP and RS.

Tables~\ref{tab:error_syn1}-\ref{tab:error_syn2} show the ME at rank-$1000$ of the sampling-based method in SYN1 and SYN2, respectively.
For both $l=1$ and $l=2$, the ME of the sampling-based method monotonically decreases when the sample size $|S_l|$ increases.
This is because, according to Theorems~\ref{thm:hoeffding} and~\ref{eq:sampleSize}, a larger sample size leads to a smaller estimation error of the pattern frequencies.

As shown in Table~\ref{tab:error_syn1}, the sampling method achieves an estimation error of $2.96\times 10^{-3}$ with a sample size $|S_l|=1,000$, which is much smaller than the \nop{lower }bound of $|S_l|$ given in Theorem~\ref{eq:sampleSize}.
The difference between the \nop{lower }bound of $|S_l|$ and the practical sample size is not really surprising, because the \nop{lower }bound in Theorem~\ref{eq:sampleSize} is \nop{a maximal lower bound of}an upper bound of sample complexity for the proposed sampling-based algorithm.

\begin{table}[t]
\caption{ME of pattern frequency in dataset SYN1.}
\centering
\label{tab:error_syn1}
\begin{tabular}{|p{1cm}<{\centering}|p{1cm}<{\centering}|p{1cm}<{\centering}|p{1cm}<{\centering}|p{1cm}<{\centering}|p{1cm}<{\centering}|} 
\hline
\multicolumn{1}{|c|}{\multirow {2}{*}{$l=1$}}& \multicolumn{1}{c|}{$|S_l|$} & \multicolumn{1}{c|}{1,000}& \multicolumn{1}{c|}{3,000}& \multicolumn{1}{c|}{5,000}& \multicolumn{1}{c|}{$8,000$}\\
\cline{2-6}
		& ME &  $2.96\times 10^{-3}$   &  $1.73\times 10^{-3}$   &  $1.33\times 10^{-3}$ &  $1.05\times 10^{-3}$\\
\hline		     
\multicolumn{1}{|c|}{\multirow {2}{*}{$l=2$}} & \multicolumn{1}{c|}{$|S_l|$} &\multicolumn{1}{c|}{10,000}& \multicolumn{1}{c|}{30,000}& \multicolumn{1}{c|}{50,000}& \multicolumn{1}{c|}{$80,000$}\\
\cline{2-6}
		& ME &  $6.52 \times 10^{-4}$   &  $3.85 \times 10^{-4}$  &  $2.88 \times 10^{-4}$ &  $2.34 \times 10^{-4}$ \\
\hline 
\end{tabular}
\end{table}

\begin{table}[t]
\caption{ME of pattern frequency in dataset SYN2.}
\centering
\label{tab:error_syn2}
\begin{tabular}{|p{1cm}<{\centering}|p{1cm}<{\centering}|p{1cm}<{\centering}|p{1cm}<{\centering}|p{1cm}<{\centering}|p{1cm}<{\centering}|} 
\hline
\multicolumn{1}{|c|}{\multirow {2}{*}{$l=1$ \centering}}& \multicolumn{1}{c|}{$|S_l|$} & \multicolumn{1}{c|}{1,000}& \multicolumn{1}{c|}{3,000}& \multicolumn{1}{c|}{5,000}& \multicolumn{1}{c|}{$8,000$}\\
\cline{2-6}
		& ME &  $3.21\times 10^{-3}$   &  $1.74\times 10^{-3}$   &  $1.49\times 10^{-3}$ &  $1.15\times 10^{-3}$ \\
\hline		     
\multicolumn{1}{|c|}{\multirow {2}{*}{$l=2$ \centering}} & \multicolumn{1}{c|}{$|S_l|$} &\multicolumn{1}{c|}{10,000}& \multicolumn{1}{c|}{30,000}& \multicolumn{1}{c|}{50,000}& \multicolumn{1}{c|}{$80,000$}\\
\cline{2-6}
		& ME &  $6.93 \times 10^{-4}$   &  $4.27 \times 10^{-4}$  &  $3.11 \times 10^{-4}$ &  $2.36 \times 10^{-4}$ \\
\hline 
\end{tabular}
\end{table}

Figures~\ref{fig:exp_syn1}(a)-(d) and~\ref{fig:exp_syn2}(a)-(d) show the AP and RS of the baseline method and the sampling-based method in SYN1 and SYN2, respectively.
The baseline always has AP and RS of $1.0$ since it produces the exact results.
For the sampling-based method, when the sample size increases, the estimation error of pattern frequency decreases. 
As a result, when the sample size is not too small, the sampling-based method can produce a stable and accurate ranked list of patterns, which is comparable with the baseline in AP and RS.

For the same path length $l$, to achieve the similar performance in AP and RS, the sampling-based method requires a larger sample size in SYN1 than in SYN2. The reason is that, we have $a=1$ in SYN1; however, in SYN2, since the numbers of transactions in different transaction databases are different, we have $a<1$. Therefore, according to Theorem~\ref{eq:sampleSize}, to achieve the same estimation quality as in SYN1, the sampling-based method needs a larger sample size in SYN2.

Figures~\ref{fig:exp_syn1}(e)-(f) and Figures~\ref{fig:exp_syn2}(e)-(f) show the runtime of the baseline method and the sampling-based method. Since the baseline method is independent from the sample size of transaction sequences, the runtime of the baseline method does not change when the sample size increases. For the sampling-based method, since a larger sample size leads to a larger time cost in transaction sequence sampling and sequential pattern mining, the runtime of the sampling-based method increases when the sample size increases.  We will test more on the scalability on the larger citation network data sets in the next subsection. %The trend is consistent with the scalability of sequential pattern mining methods~\cite{Pei2001}.

In summary, on the synthetic datasets the sampling-based method significantly reduces the number of required transaction sequences and achieves good approximation quality.

\newcommand{\parawidth}{45mm}

\begin{figure}[t]
\centering
	\subfigure[AP, SYN1, $l=1$]{\includegraphics[width=\parawidth]{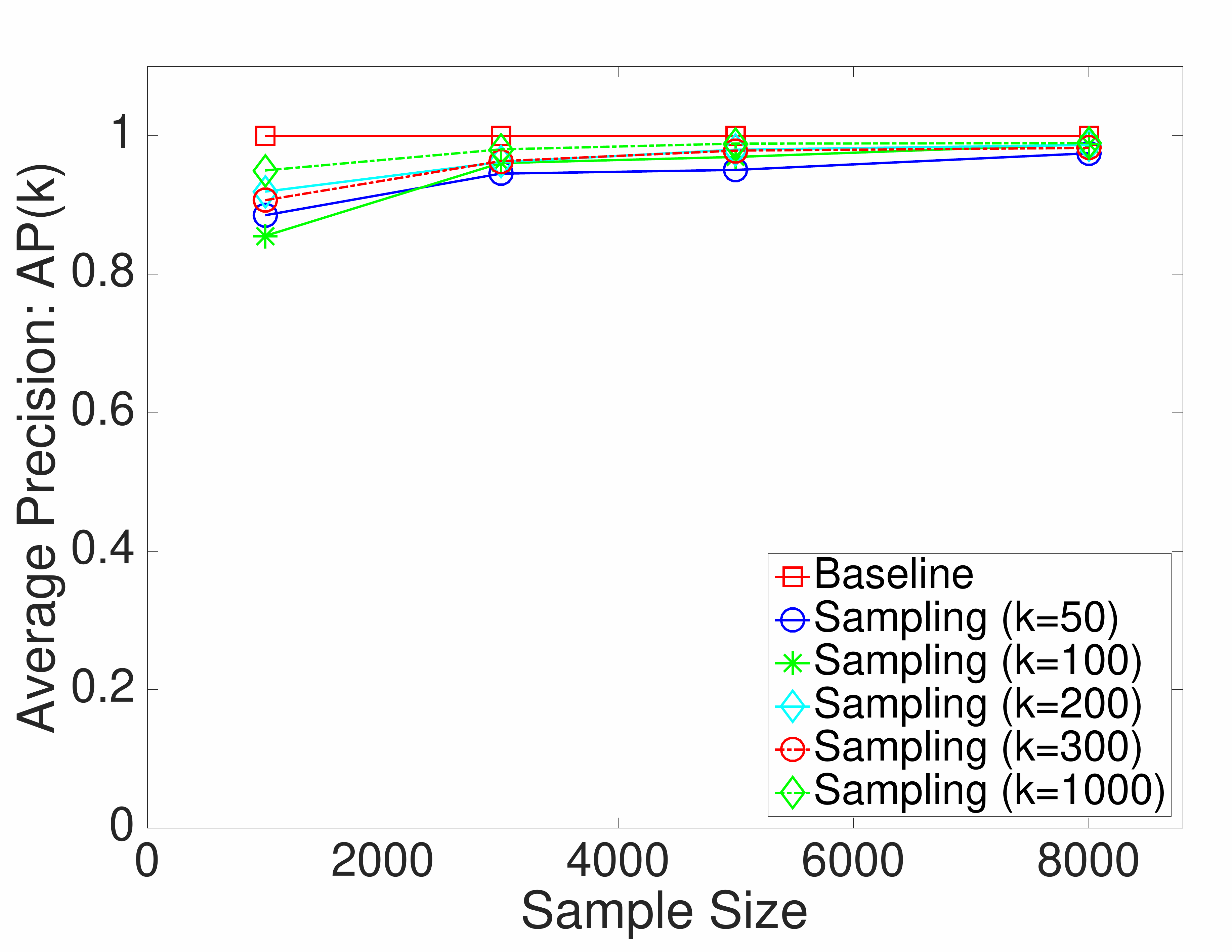}}
	\subfigure[AP, SYN1, $l=2$]{\includegraphics[width=\parawidth]{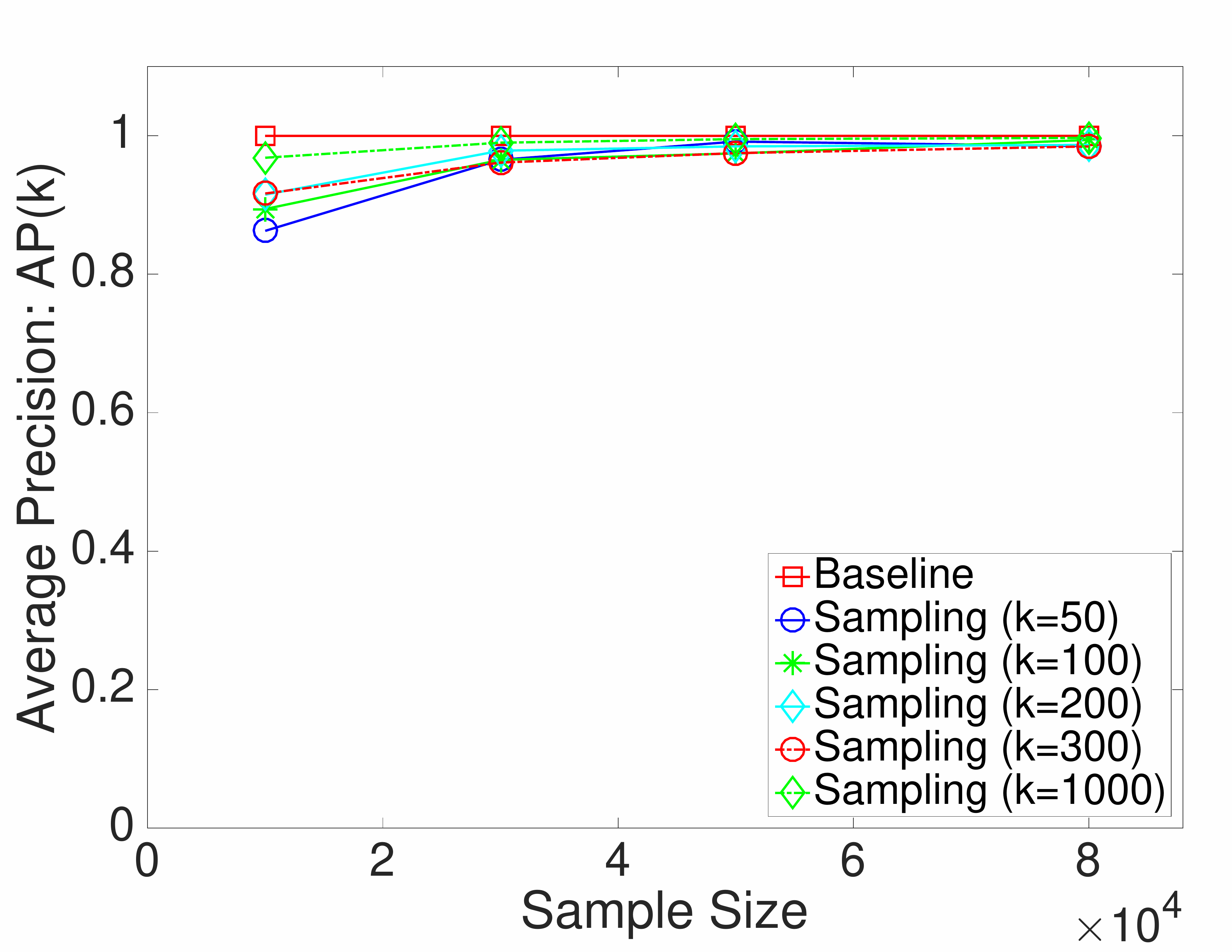}}

	\subfigure[RS, SYN1, $l=1$]{\includegraphics[width=\parawidth]{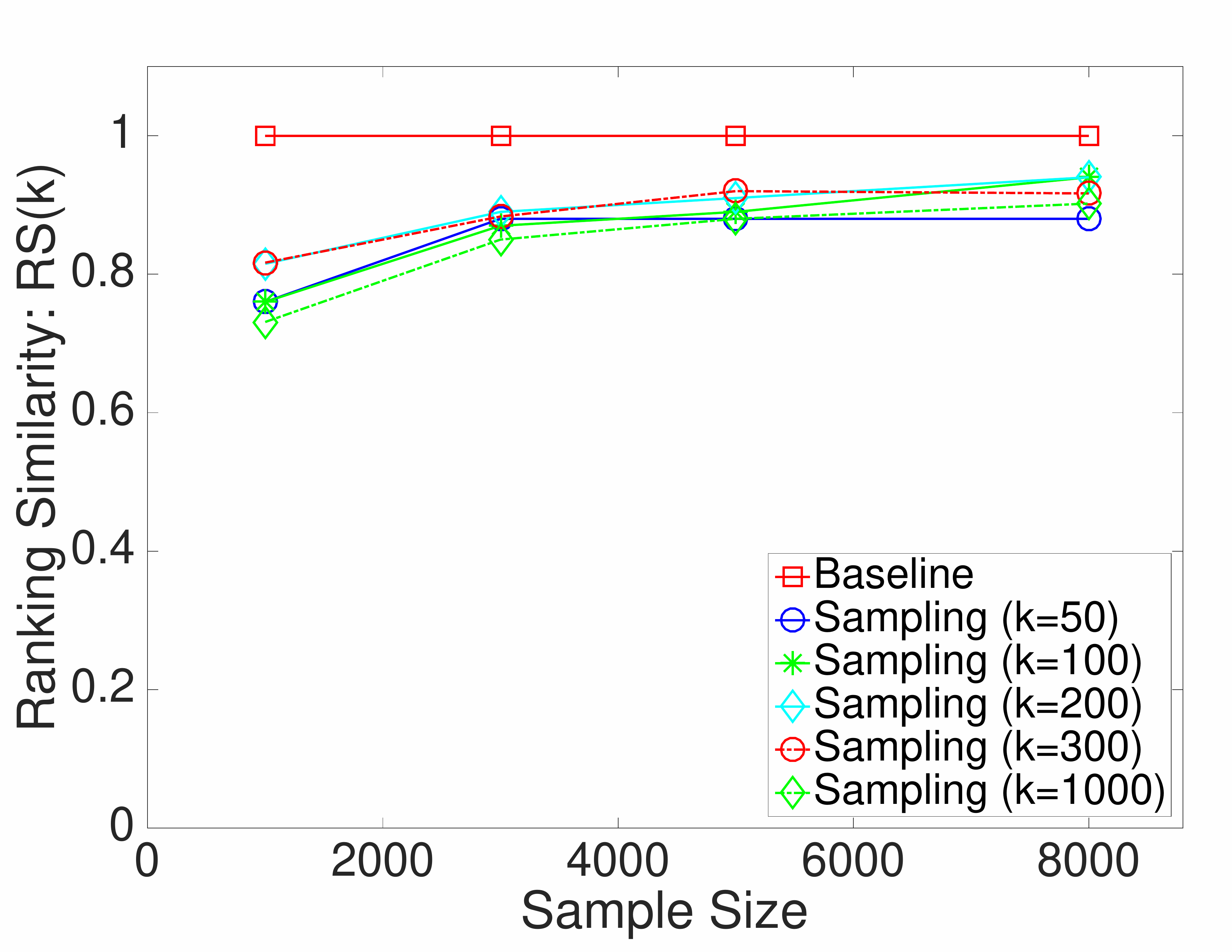}}
	\subfigure[RS, SYN1, $l=2$]{\includegraphics[width=\parawidth]{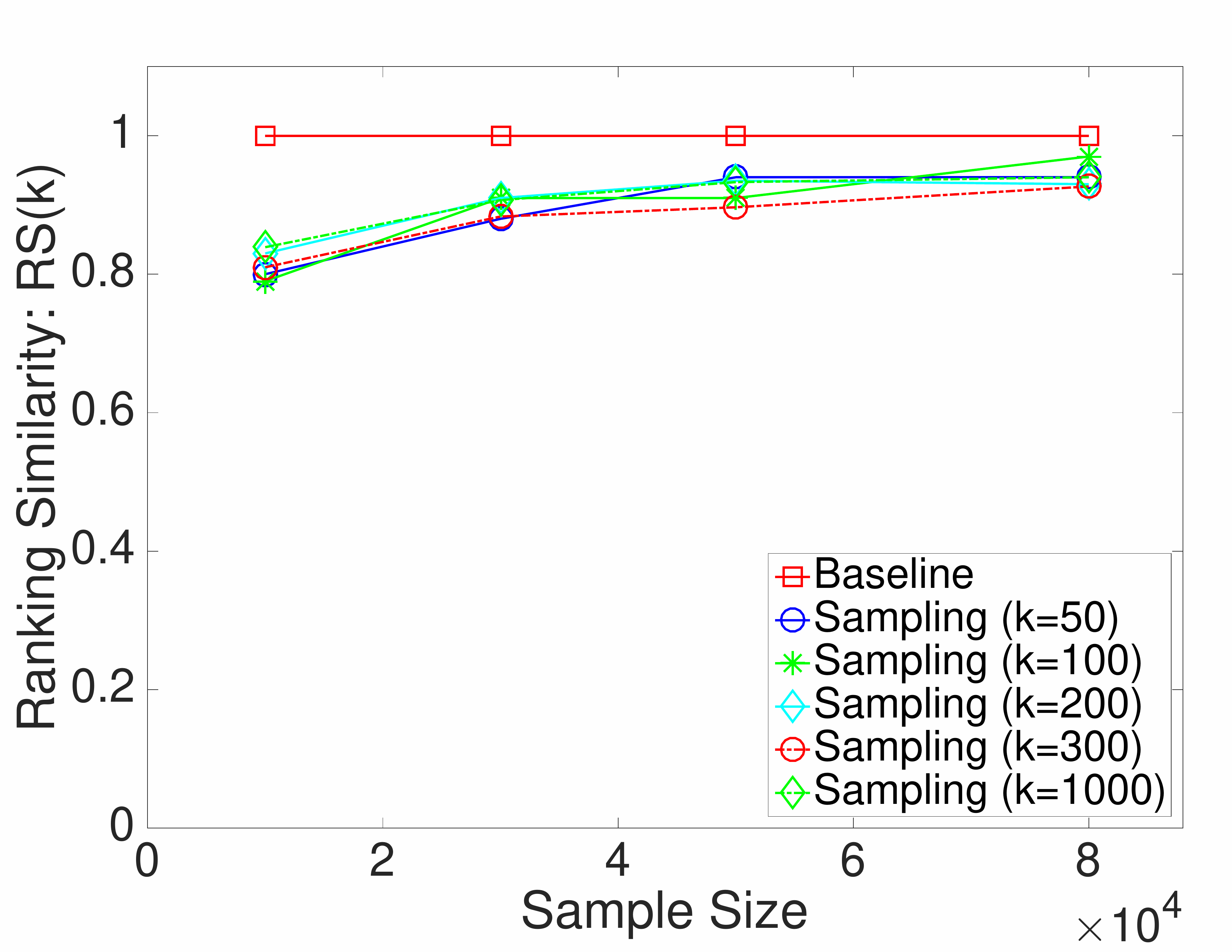}}
	
	\subfigure[RT, SYN1, $l=1$]{\includegraphics[width=\parawidth]{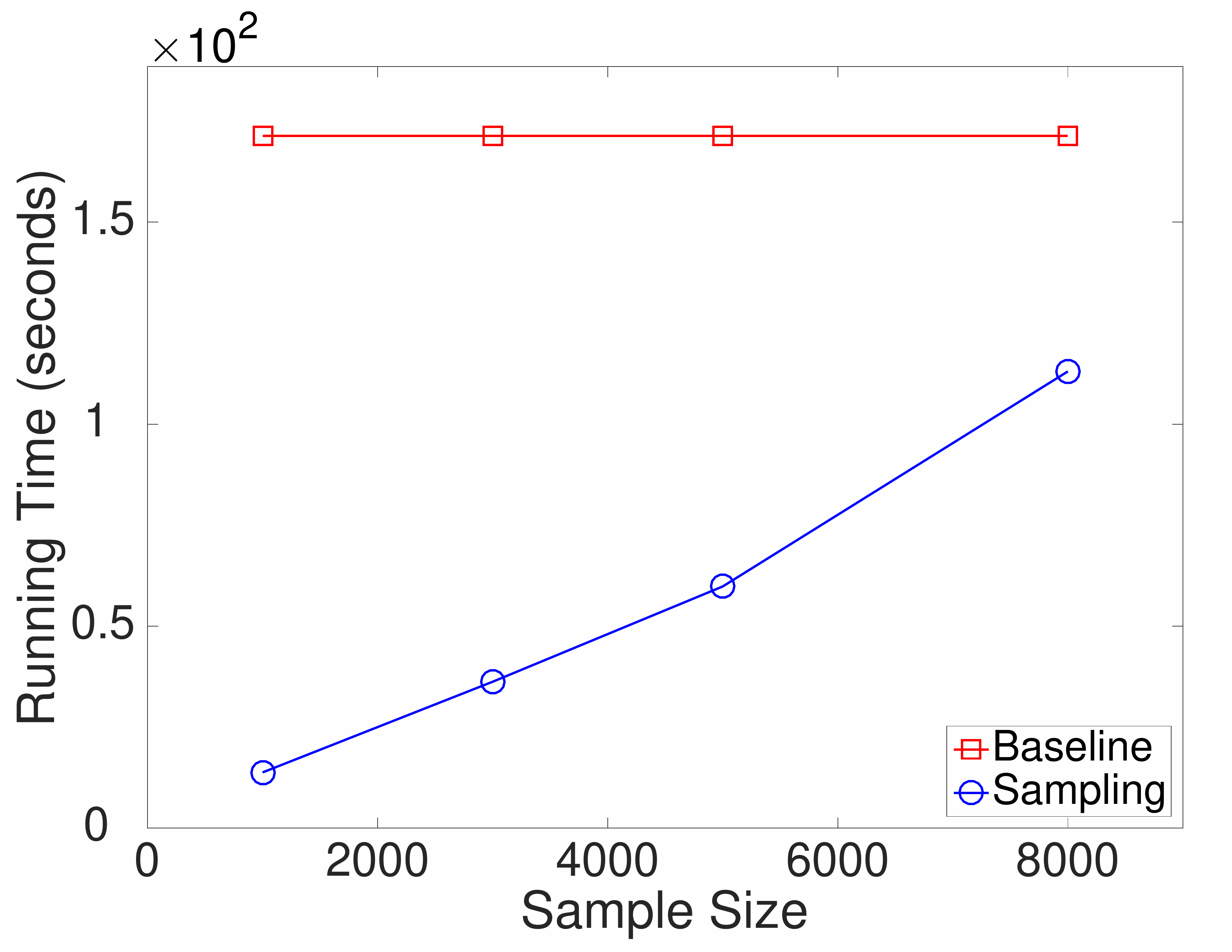}}
	\subfigure[RT, SYN1, $l=2$]{\includegraphics[width=\parawidth]{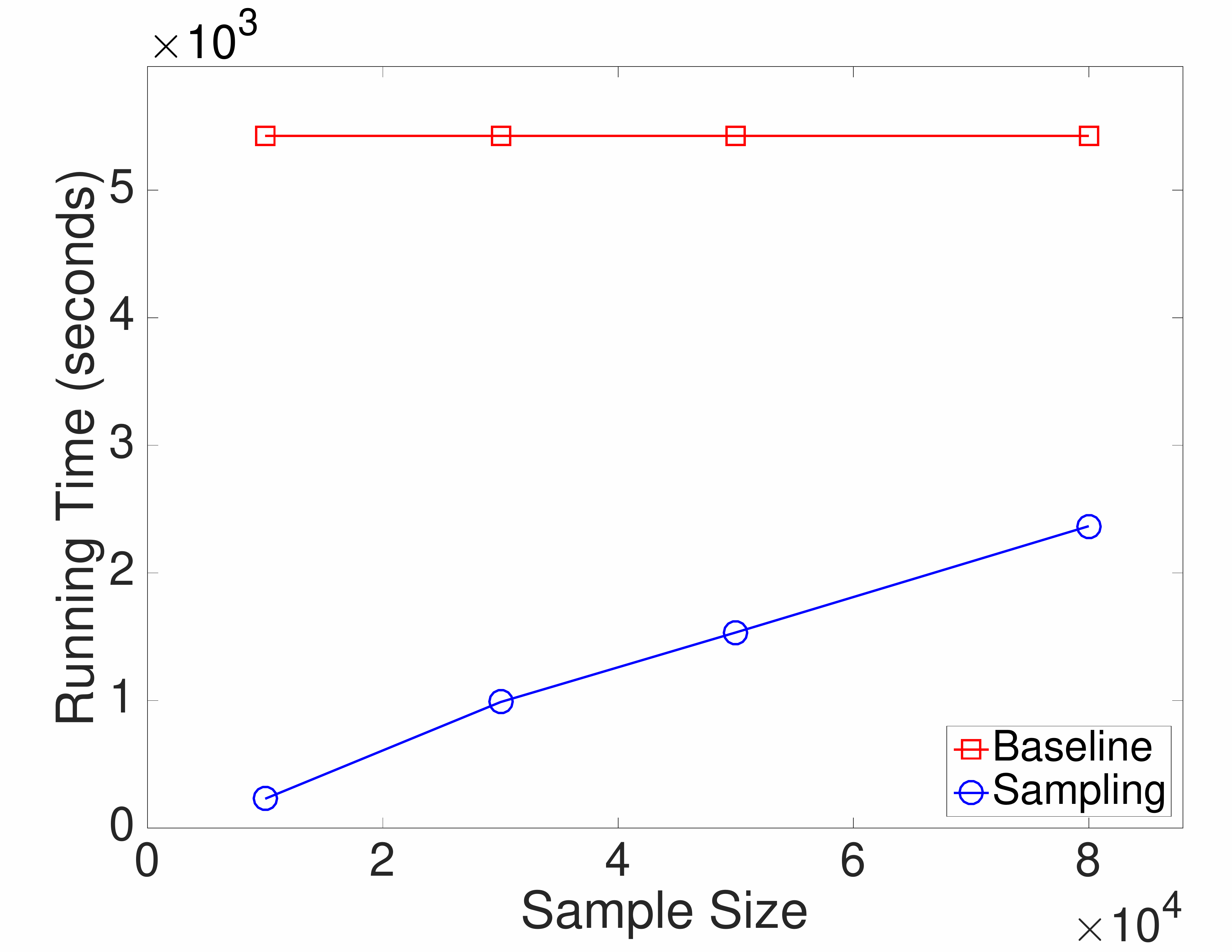}}
\caption{The AP, RS and runtime of the sampling-based method and the baseline method on the SYN1 dataset. The sample sizes are $|S_l|=1000, 3000, 5000, 8000$.}
\label{fig:exp_syn1} 
\end{figure}

\begin{figure}[t]
\centering
	\subfigure[AP, SYN2, $l=1$]{\includegraphics[width=\parawidth]{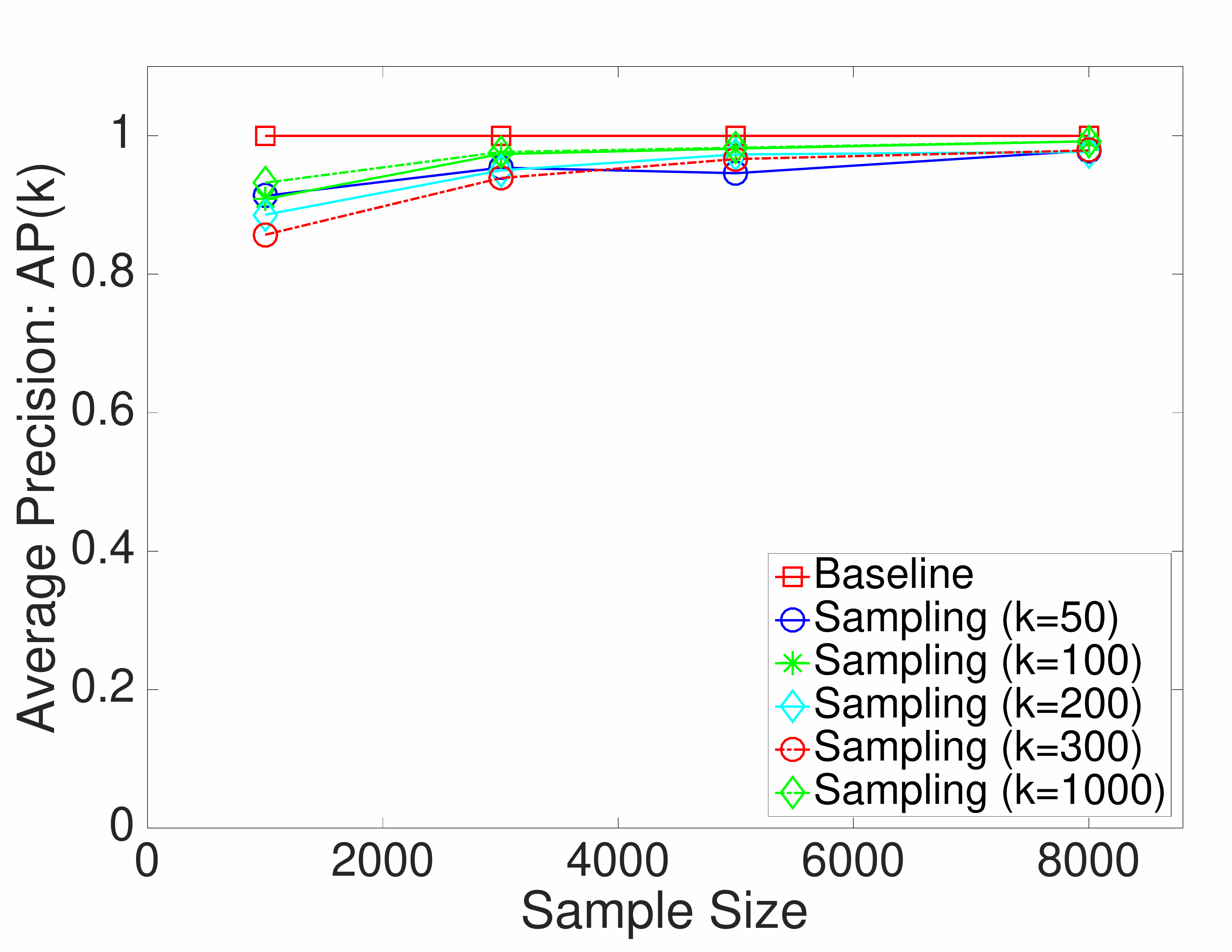}}
	\subfigure[AP, SYN2, $l=2$]{\includegraphics[width=\parawidth]{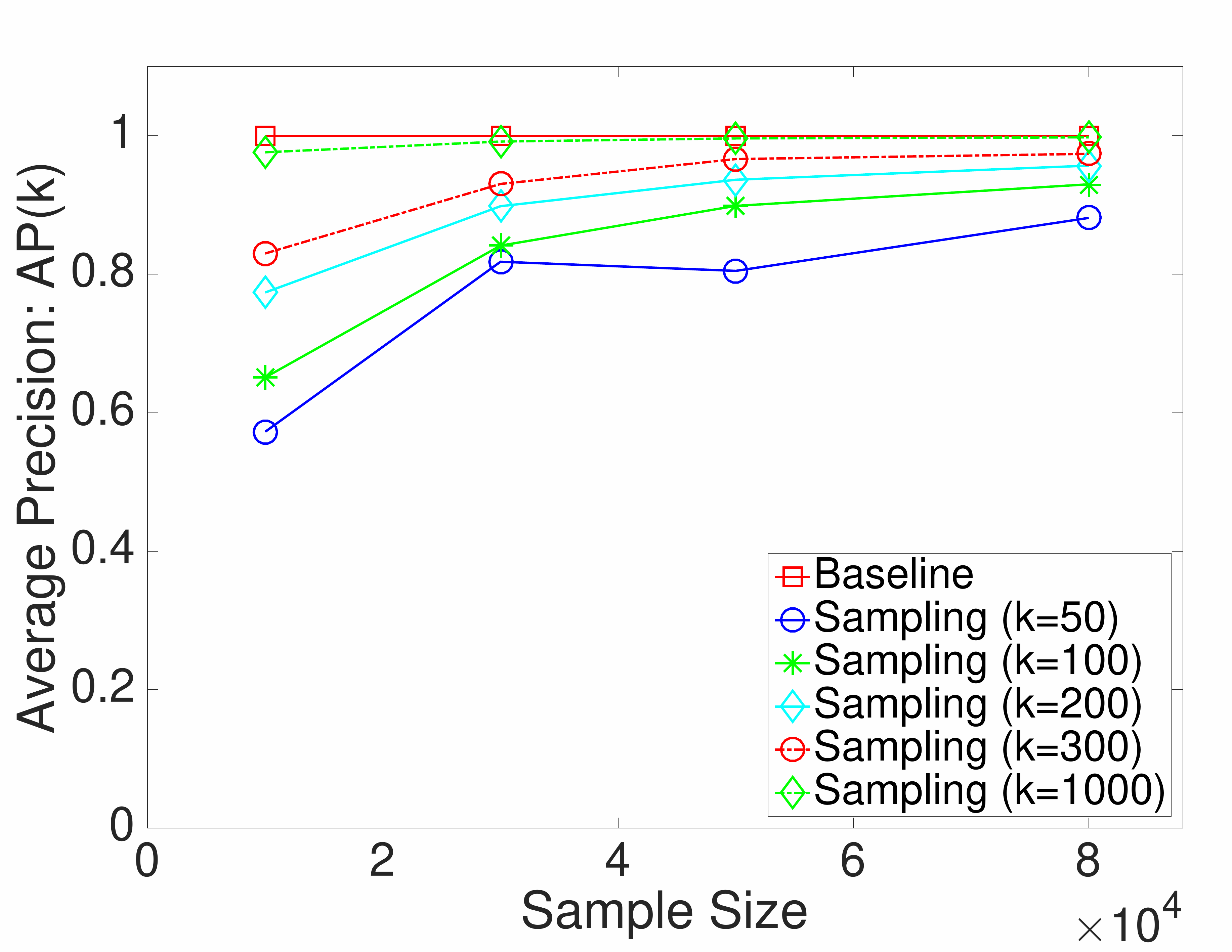}}
	
	\subfigure[RS, SYN2, $l=1$]{\includegraphics[width=\parawidth]{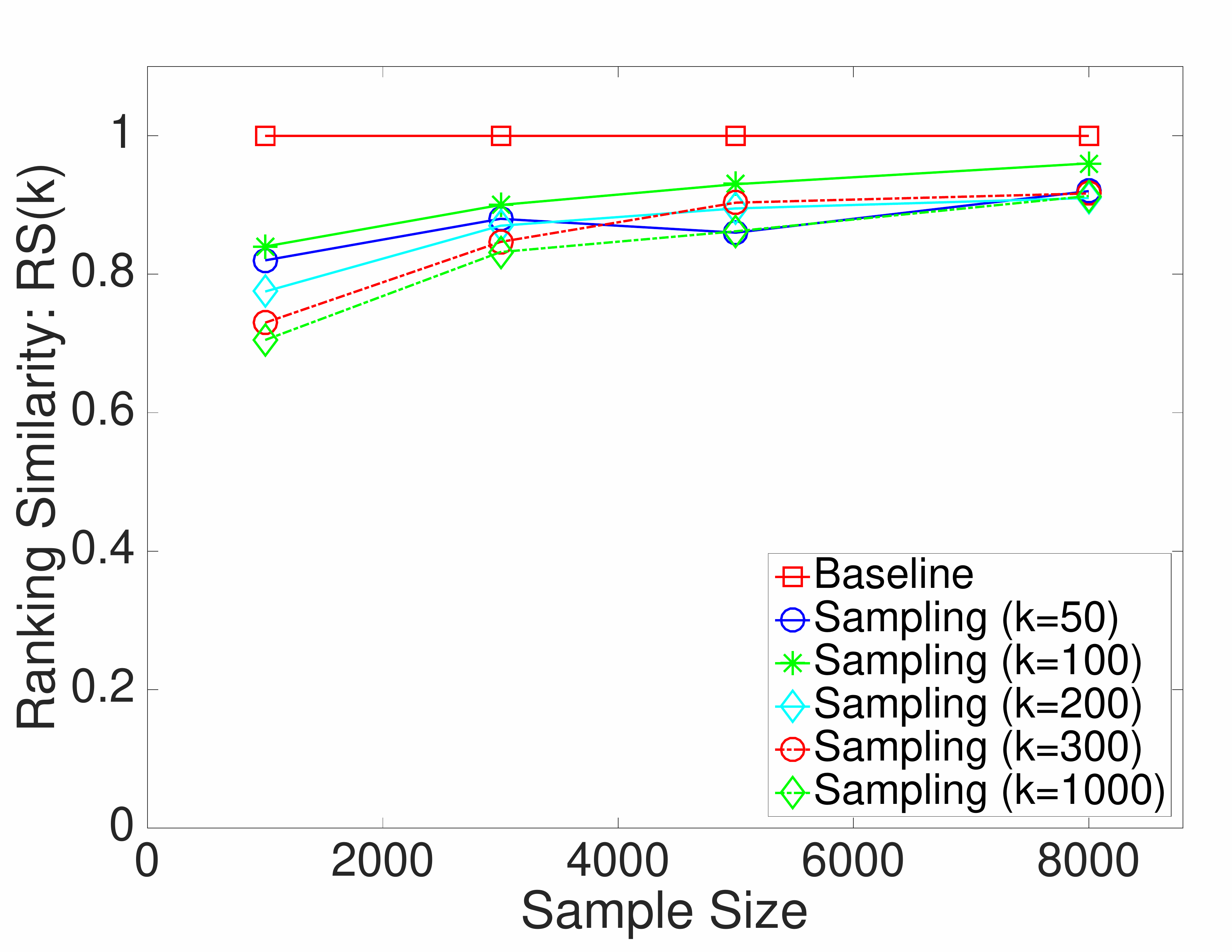}}
	\subfigure[RS, SYN2, $l=2$]{\includegraphics[width=\parawidth]{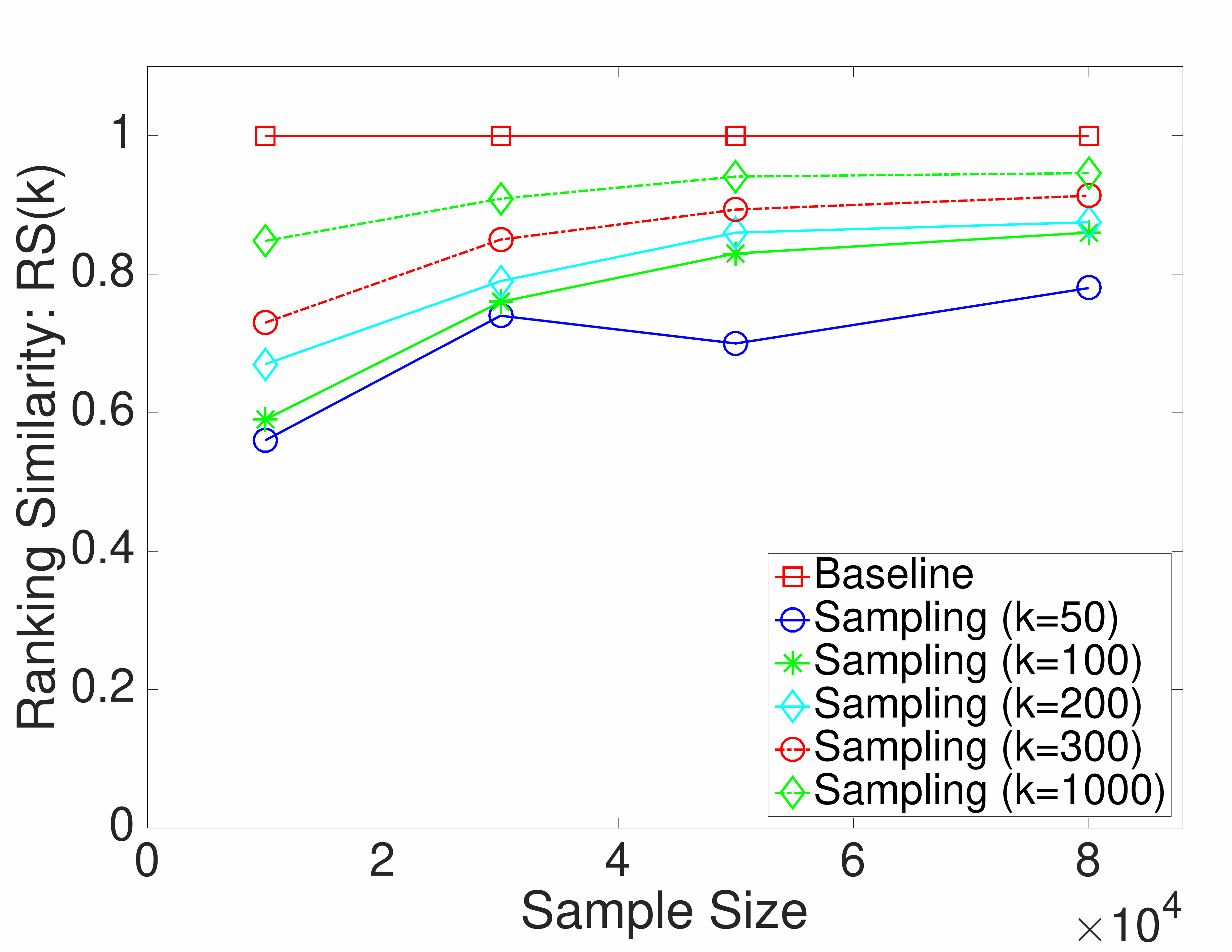}}
	
	\subfigure[RT, SYN2, $l=1$]{\includegraphics[width=\parawidth]{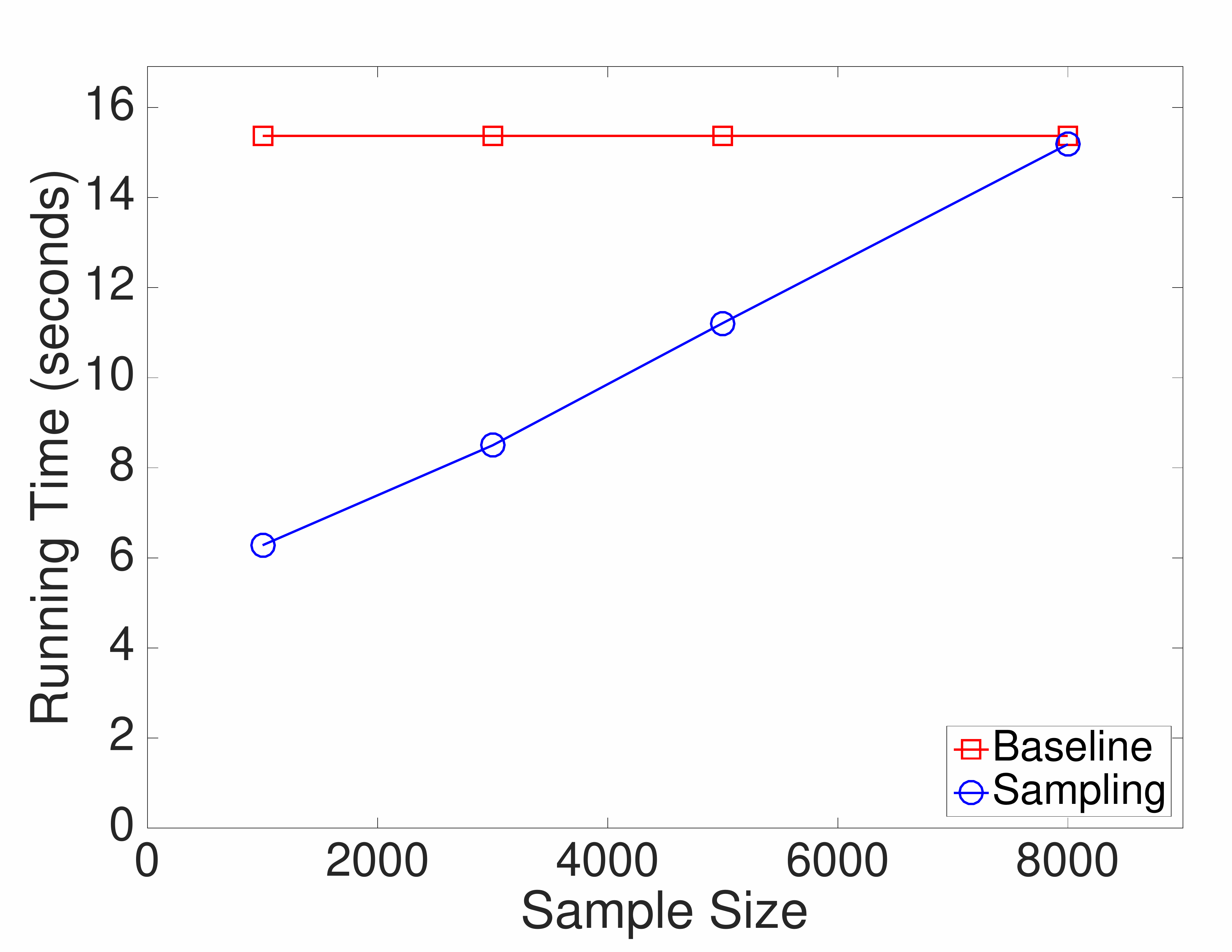}}	
	\subfigure[RT, SYN2, $l=2$]{\includegraphics[width=\parawidth]{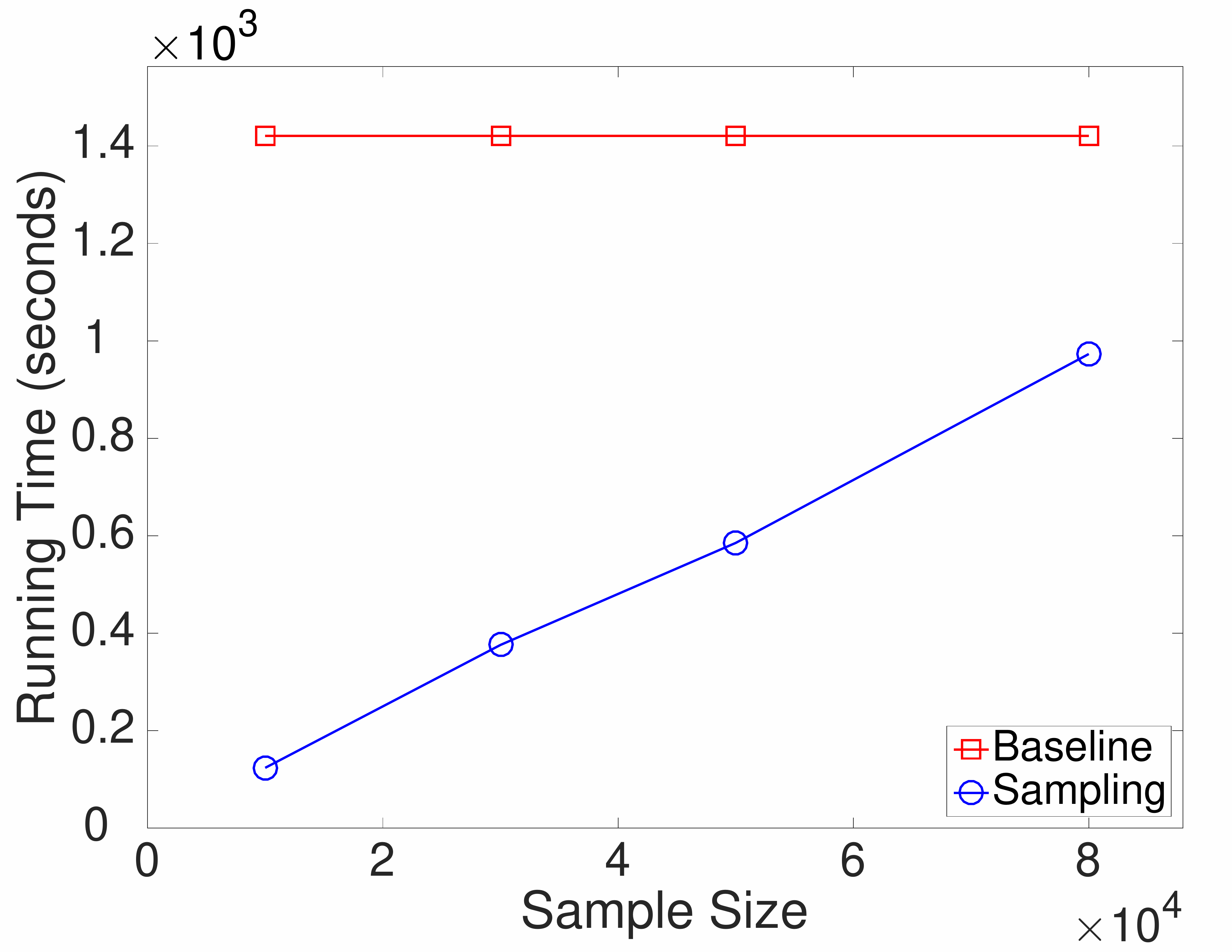}}
\caption{The AP, RS and runtime of the sampling-based method and the baseline method on the SYN2 dataset. The sample sizes are $|S_l|=1\times 10^4, 3\times 10^4, 5\times 10^4, 8\times 10^4$.}
\label{fig:exp_syn2} 
\end{figure}

\begin{figure}[t]
\centering
	\subfigure[AP, Flight, $l=1$]{\includegraphics[width=\parawidth]{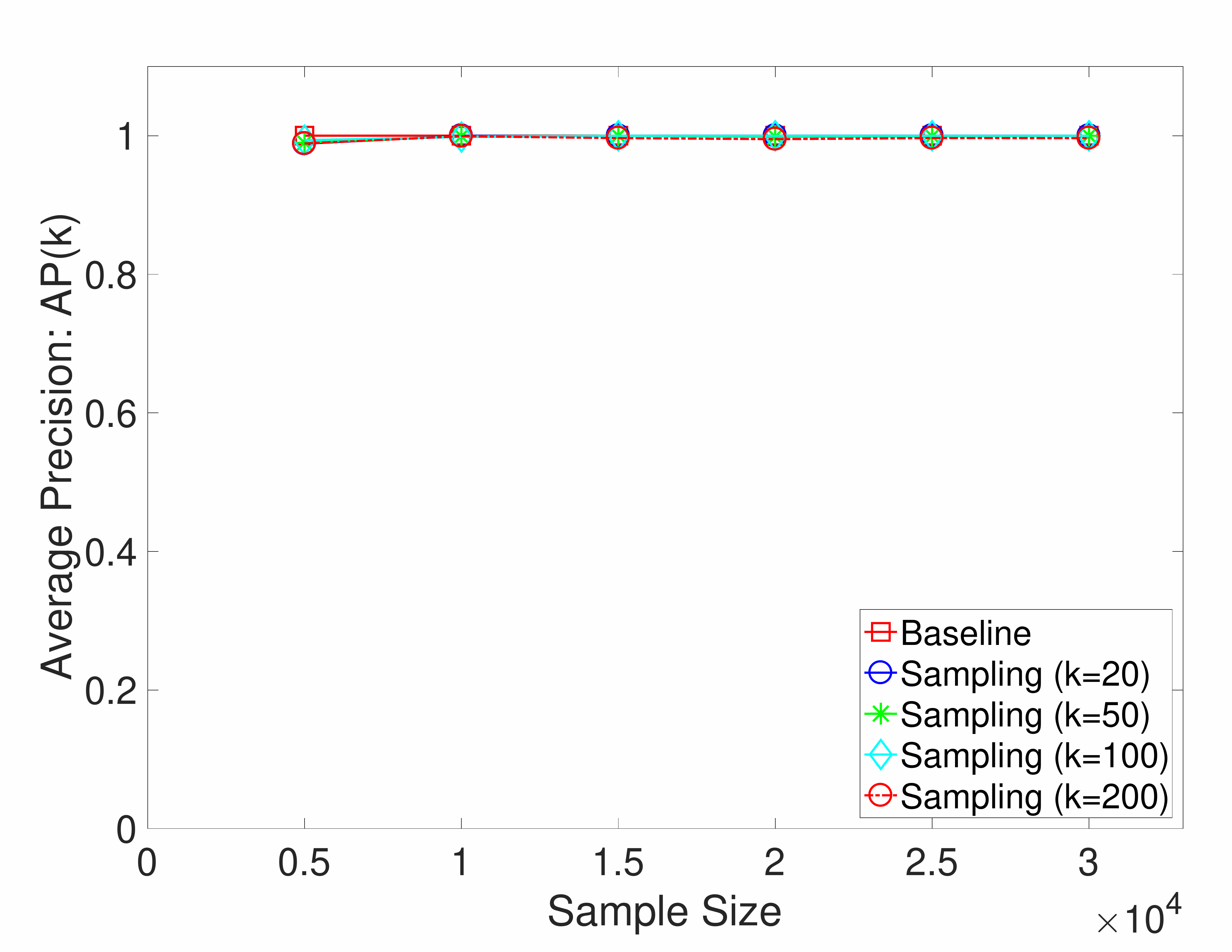}}
	\subfigure[AP, Flight, $l=2$]{\includegraphics[width=\parawidth]{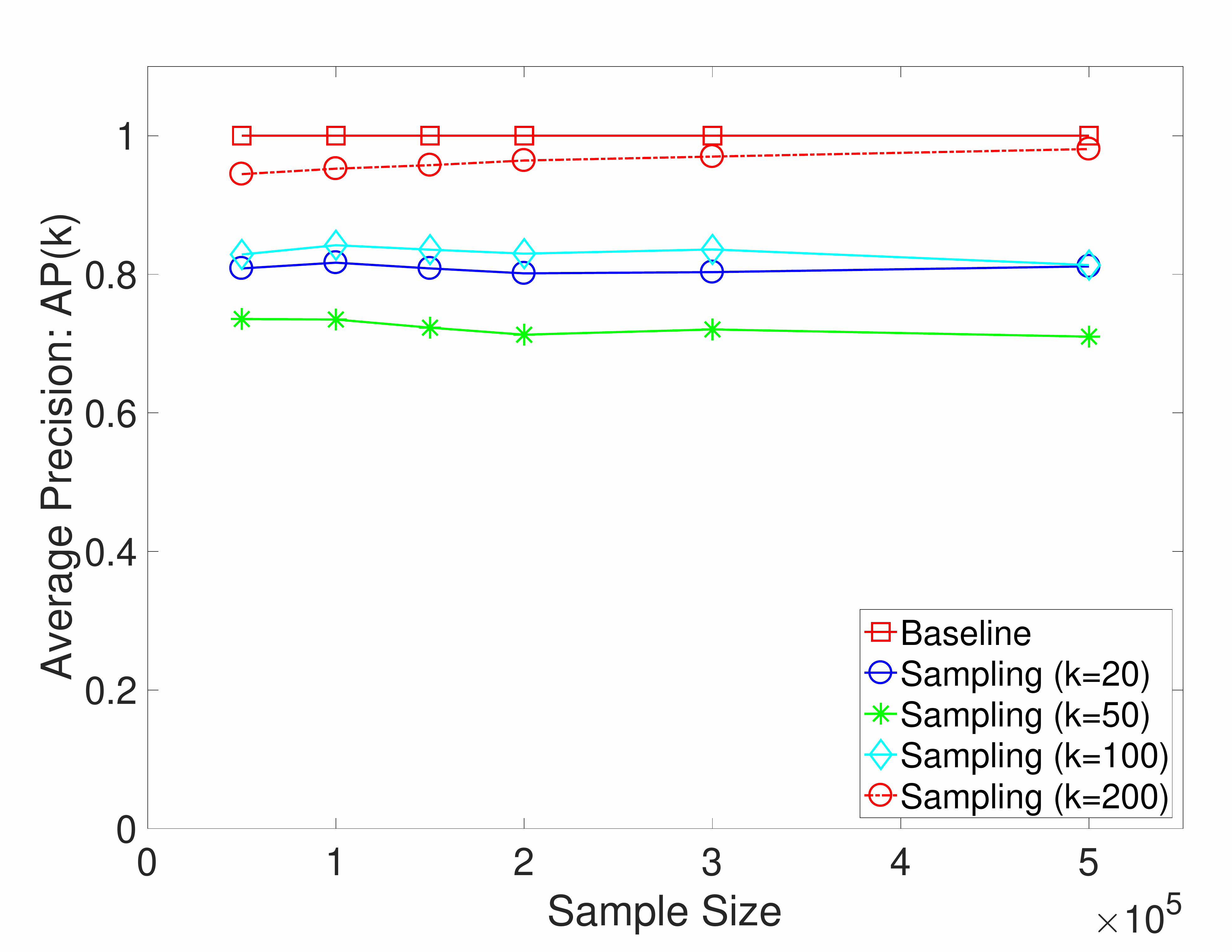}}
	
	\subfigure[RS, Flight, $l=1$]{\includegraphics[width=\parawidth]{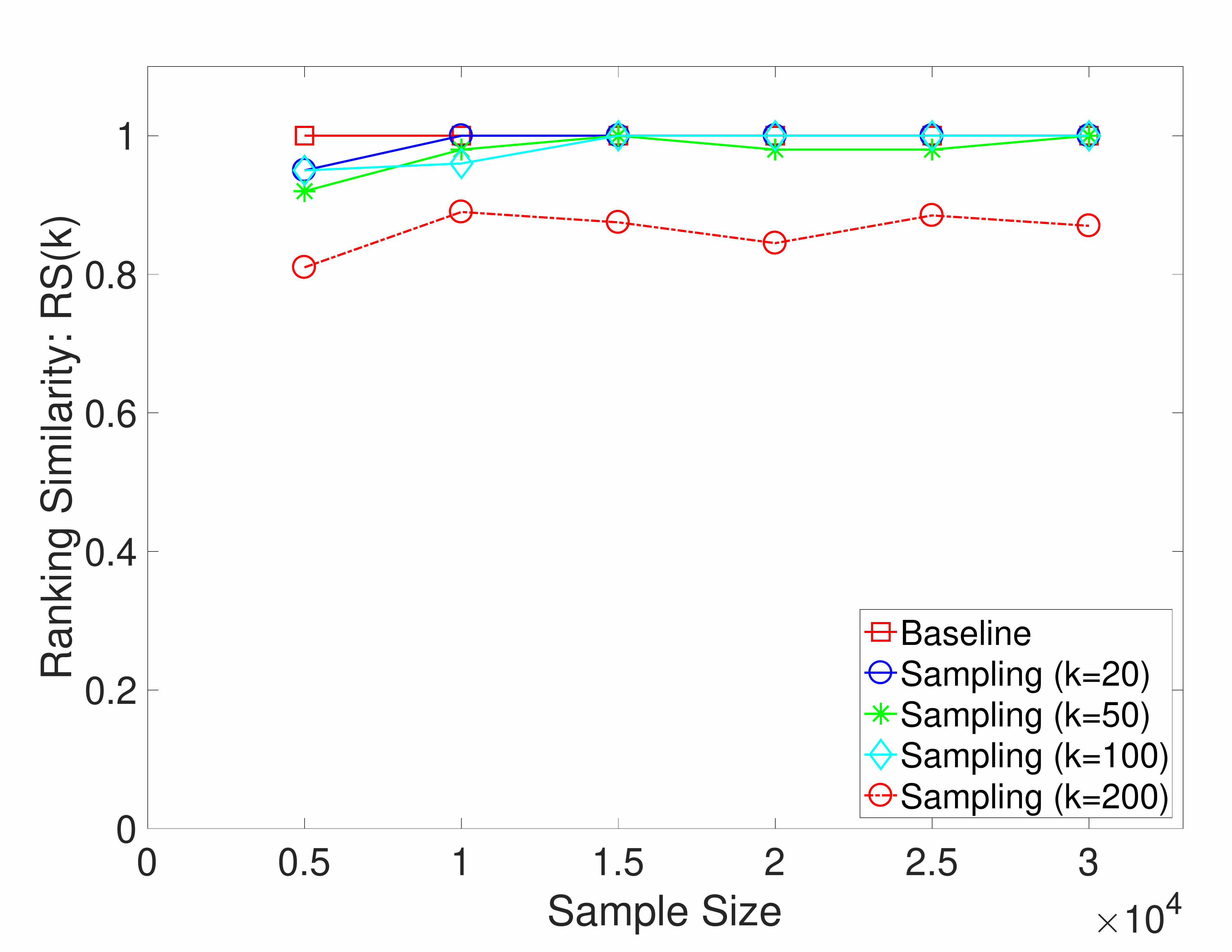}}
	\subfigure[RS, Flight, $l=2$]{\includegraphics[width=\parawidth]{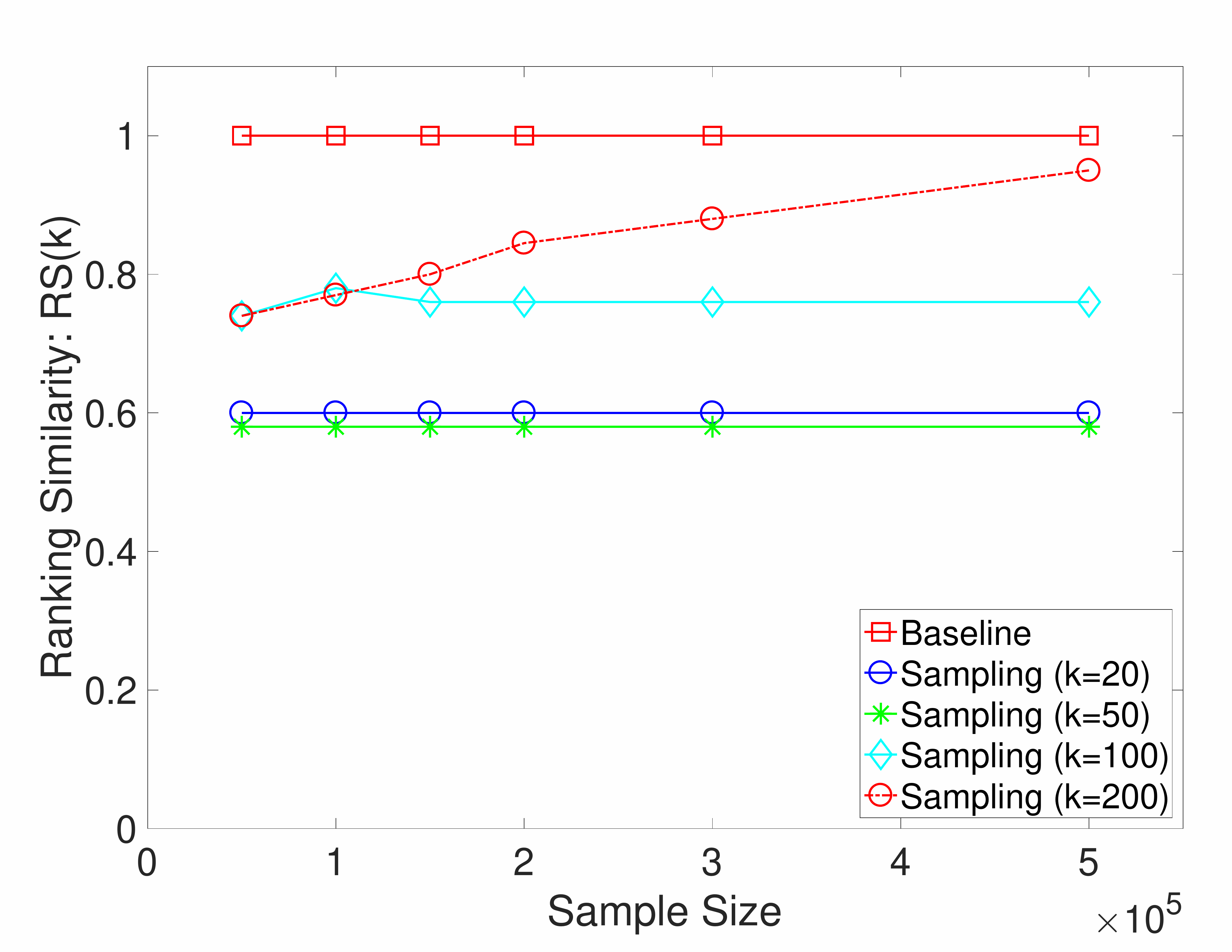}}

	\subfigure[RT, Flight, $l=1$]{\includegraphics[width=\parawidth]{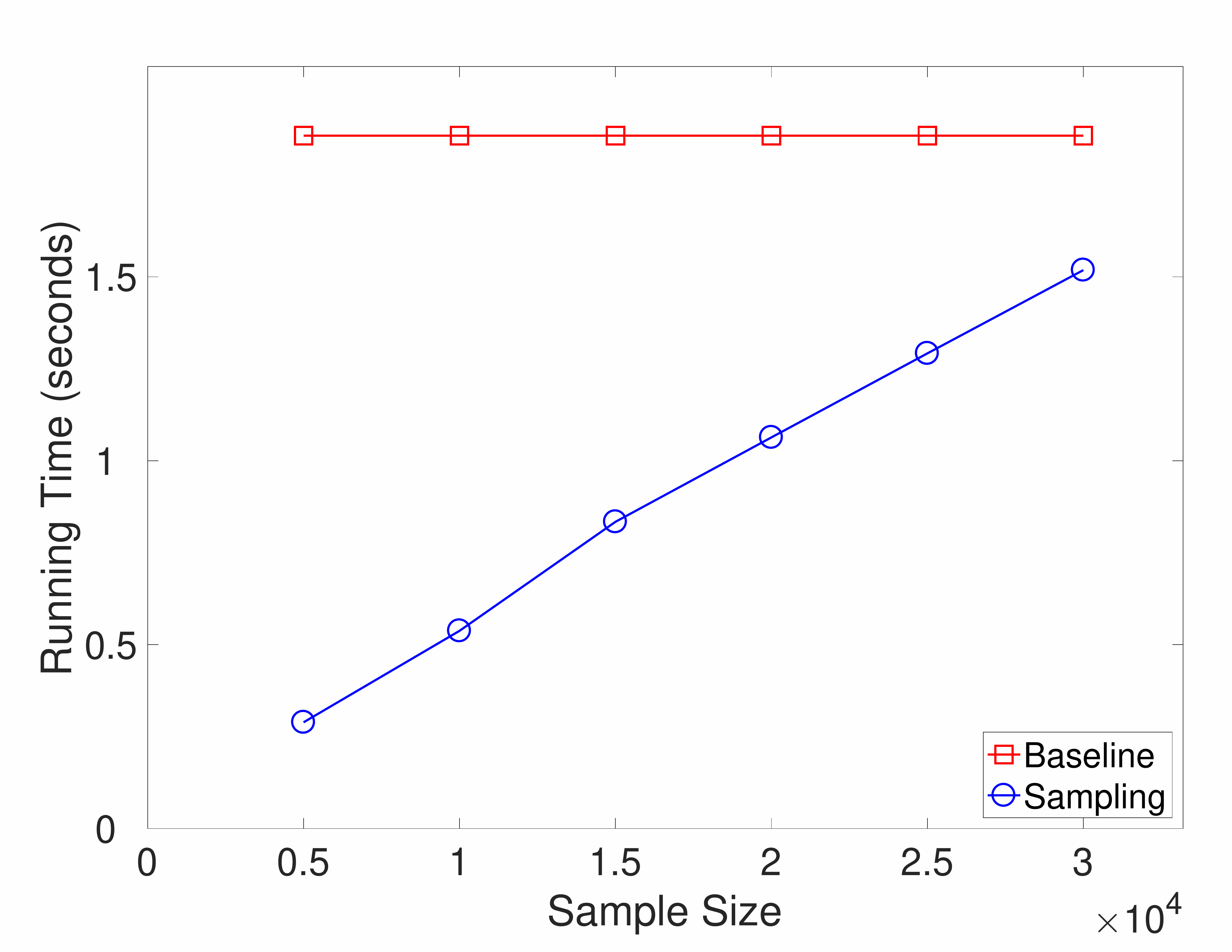}}
	\subfigure[RT, Flight, $l=2$]{\includegraphics[width=\parawidth]{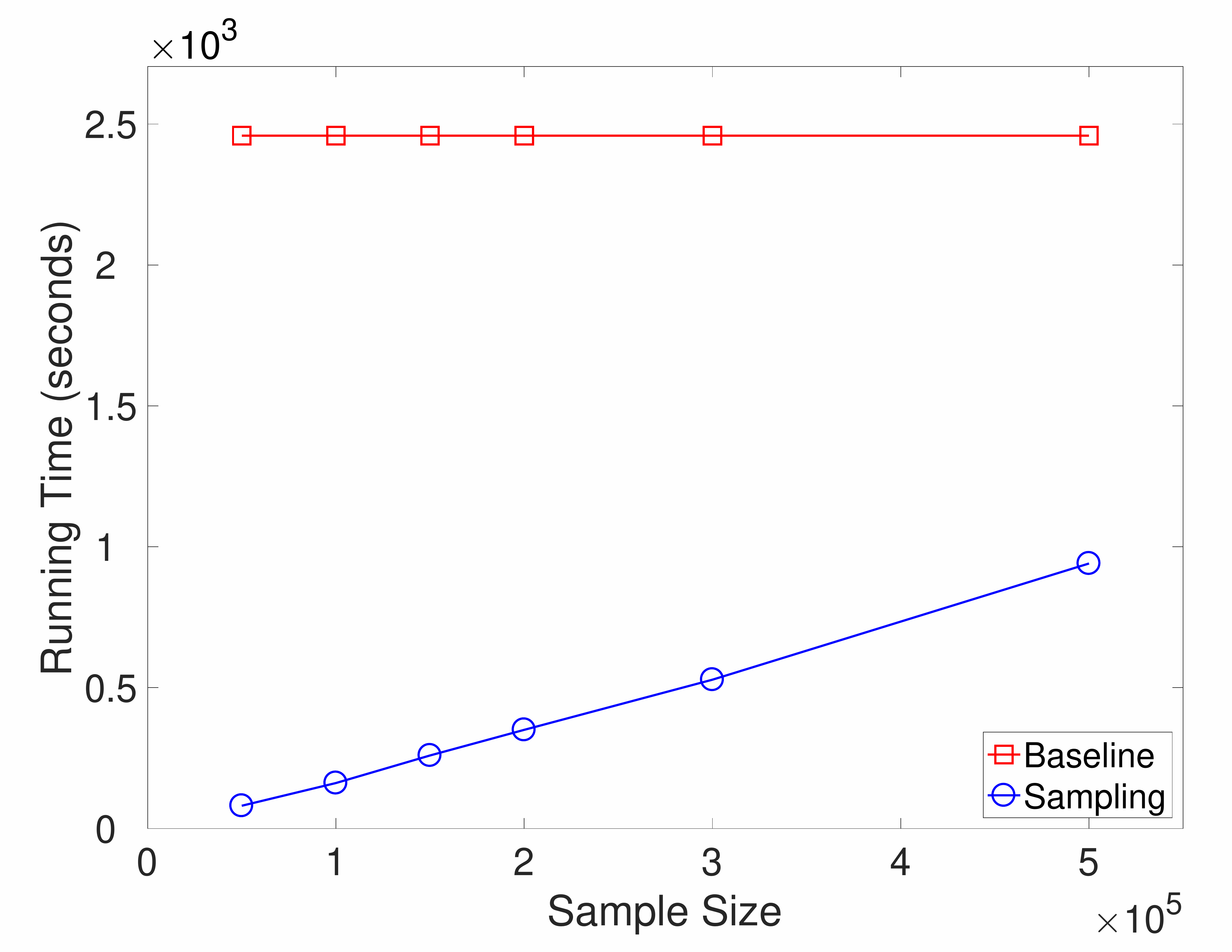}}
\caption{The AP, RS and runtime of the sampling-based method on the Flight dataset.}
\label{fig:flight} 
\end{figure}

\begin{figure}[t]
\centering
	\subfigure[AP, CN1, $l=1$]{\includegraphics[width=\parawidth]{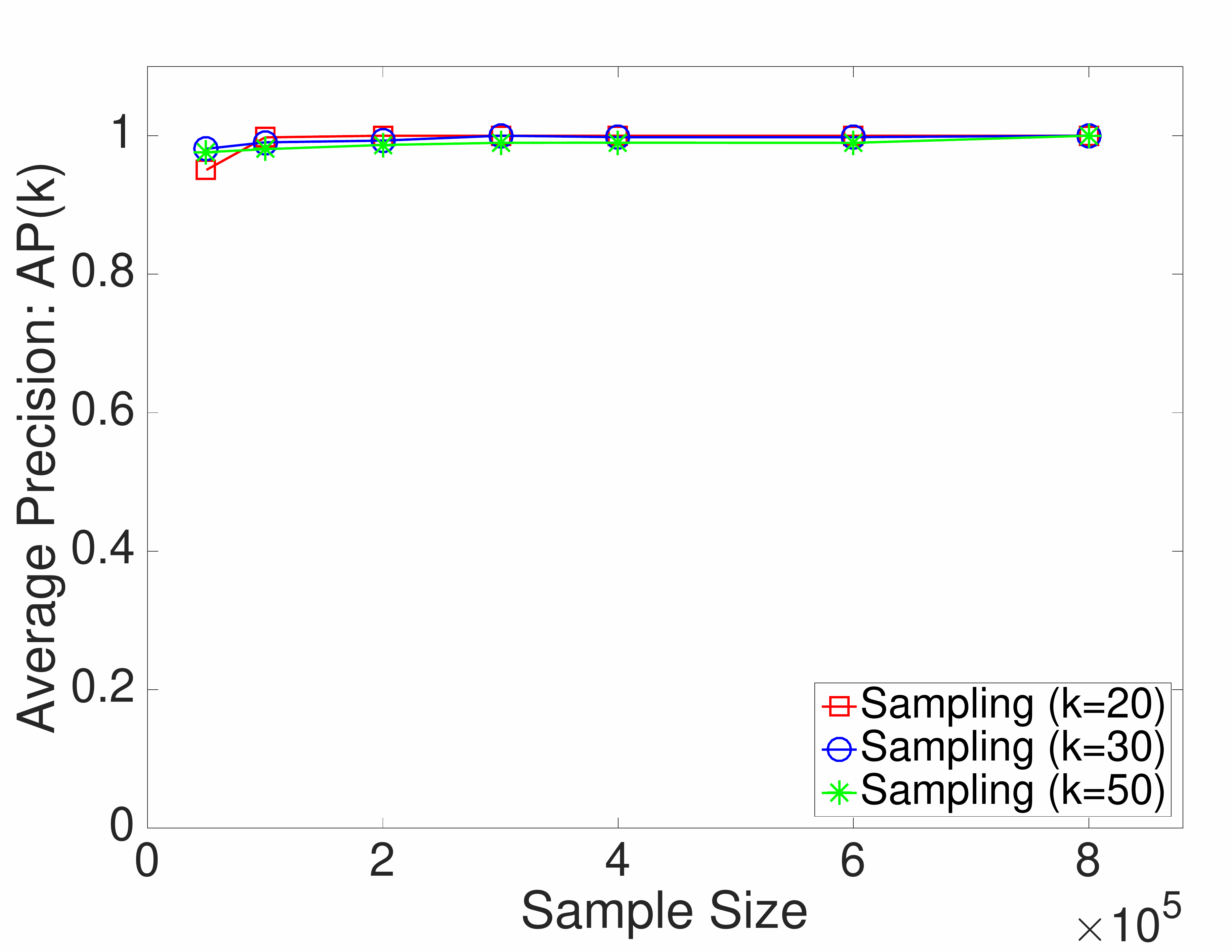}}
	\subfigure[AP, CN1, $l=2$]{\includegraphics[width=\parawidth]{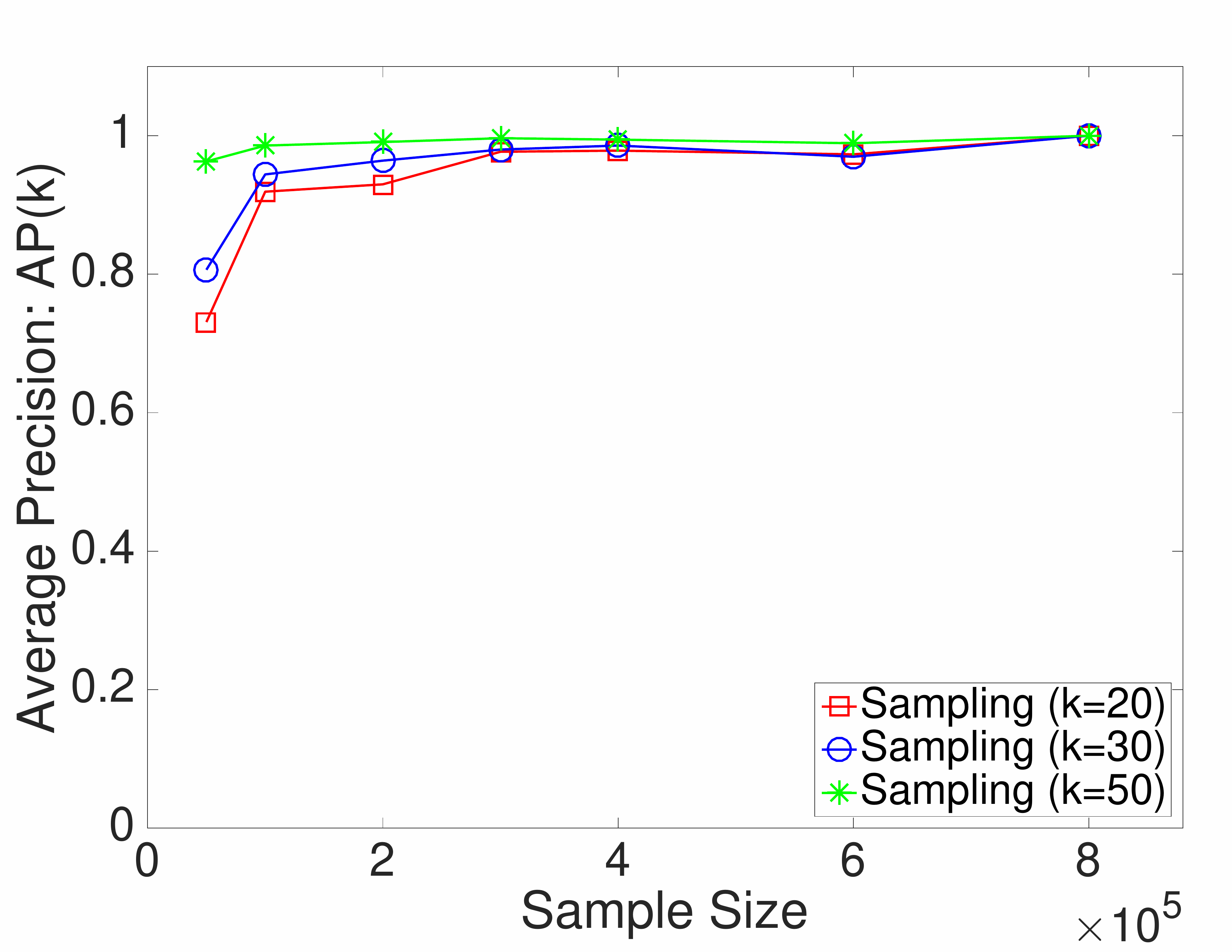}}
	
	\subfigure[RS, CN1, $l=1$]{\includegraphics[width=\parawidth]{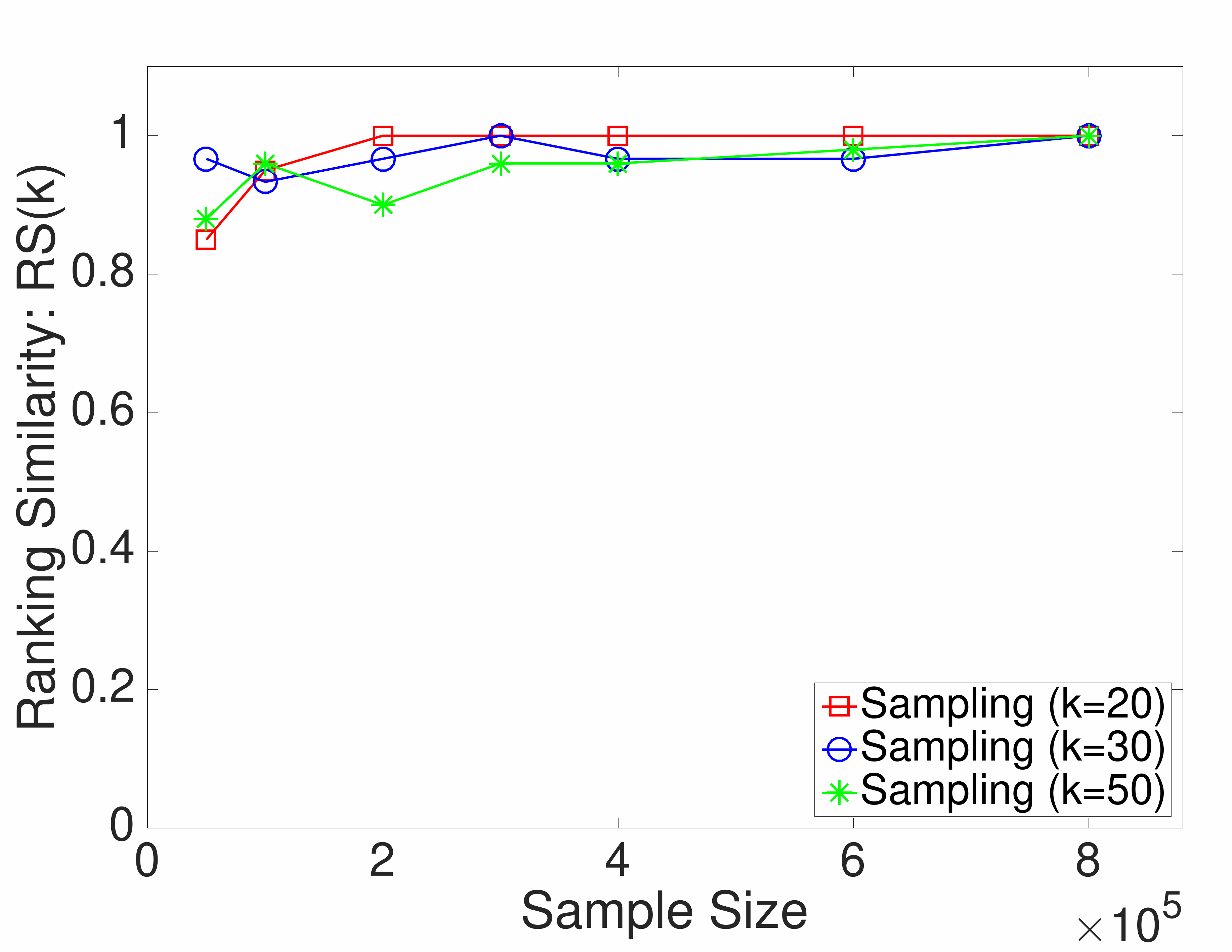}}
	\subfigure[RS, CN1, $l=2$]{\includegraphics[width=\parawidth]{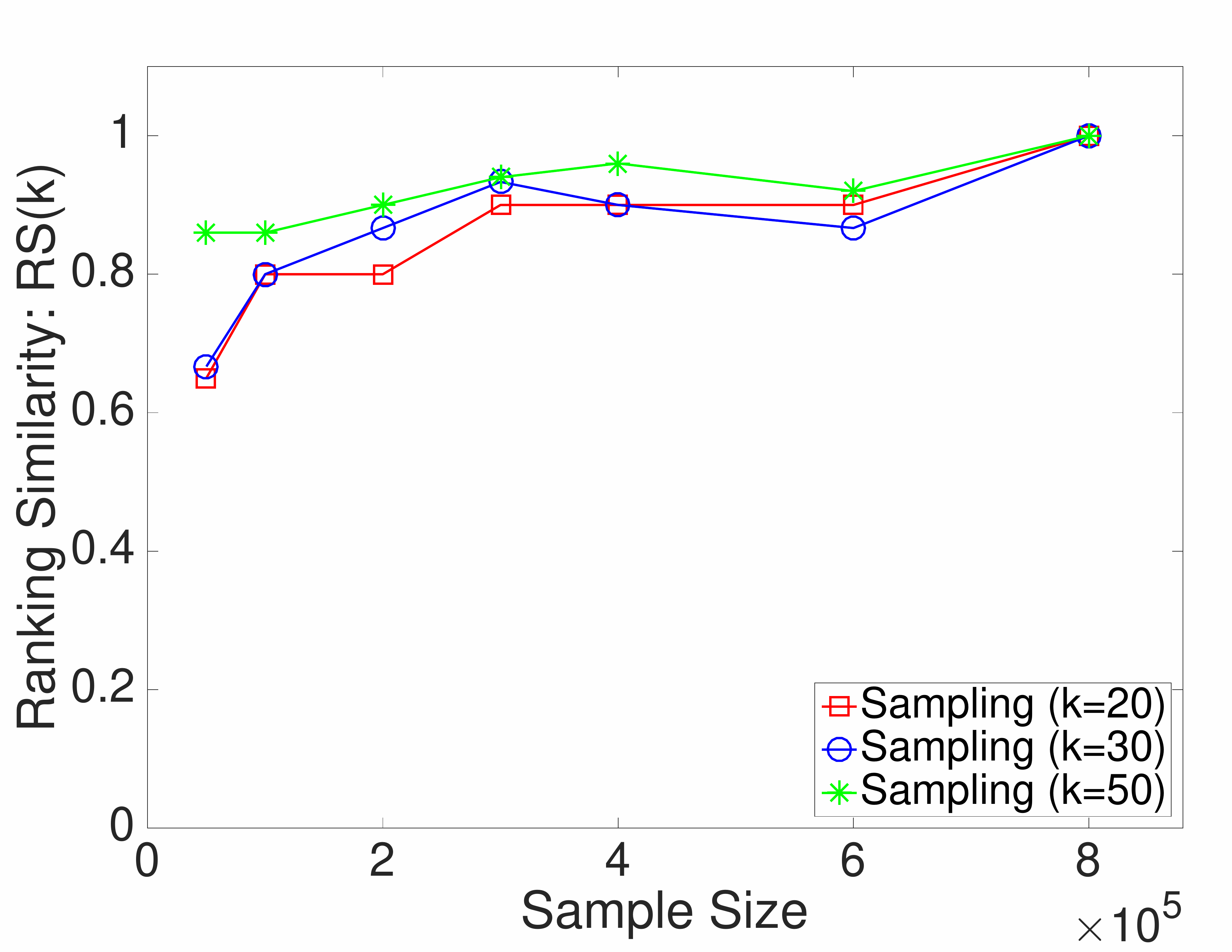}}

	\subfigure[RT, CN1, $l=1$]{\includegraphics[width=\parawidth]{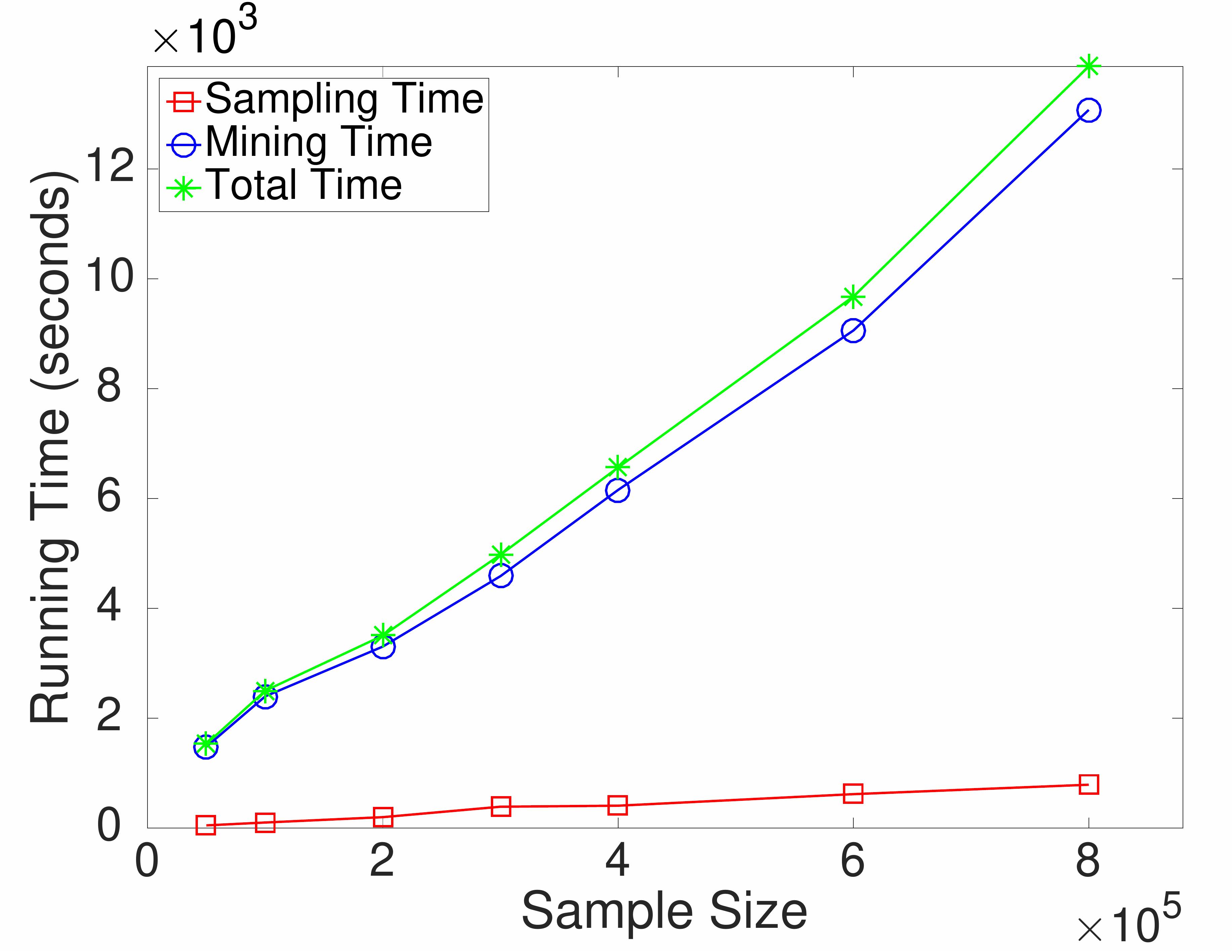}}
	\subfigure[RT, CN1, $l=2$]{\includegraphics[width=\parawidth]{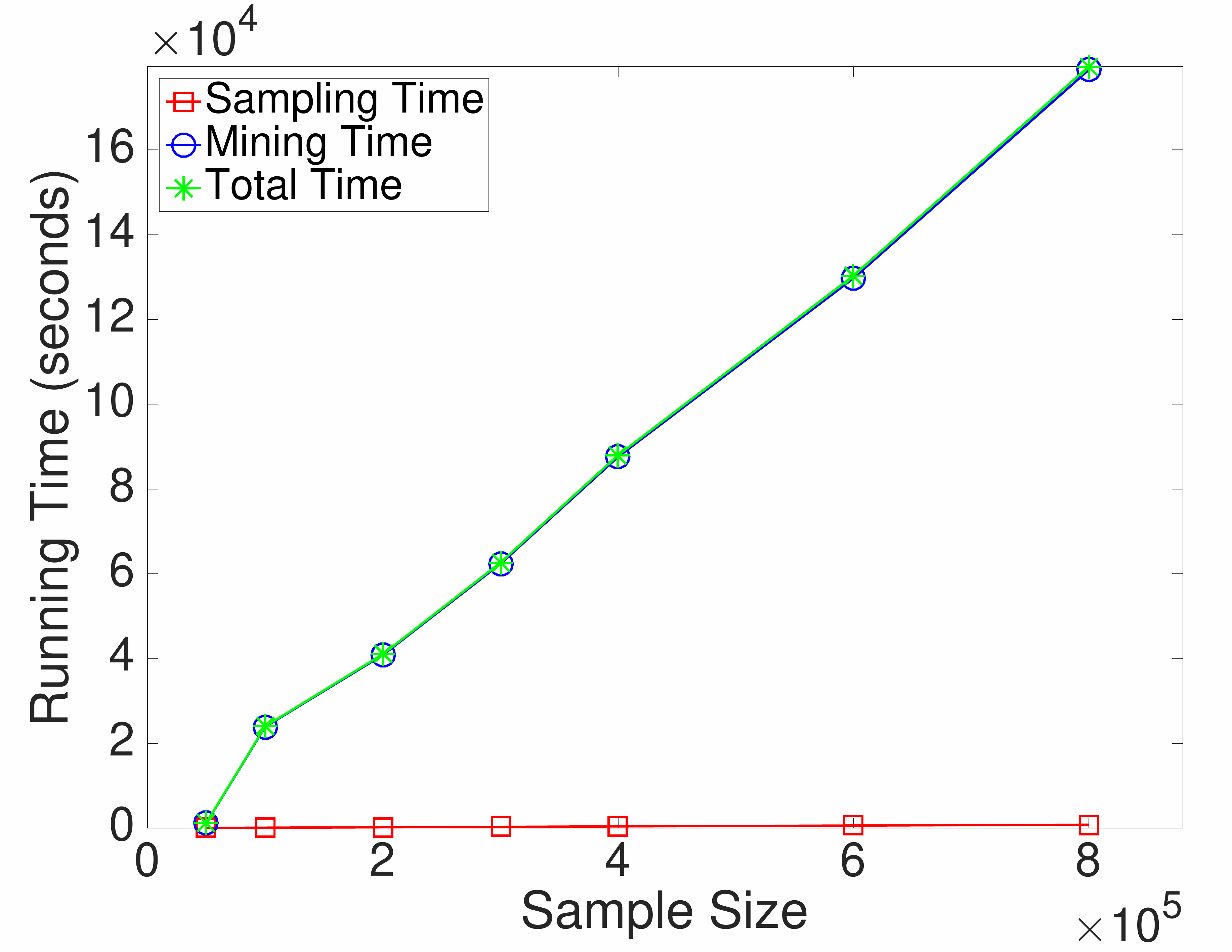}}
\caption{The AP, RS and runtime of the sampling-based method on the CN1 dataset.}
\label{fig:exp_cn1} 
\end{figure}

\begin{figure}[t]
\centering
	\subfigure[AP, CN2, $l=1$]{\includegraphics[width=\parawidth]{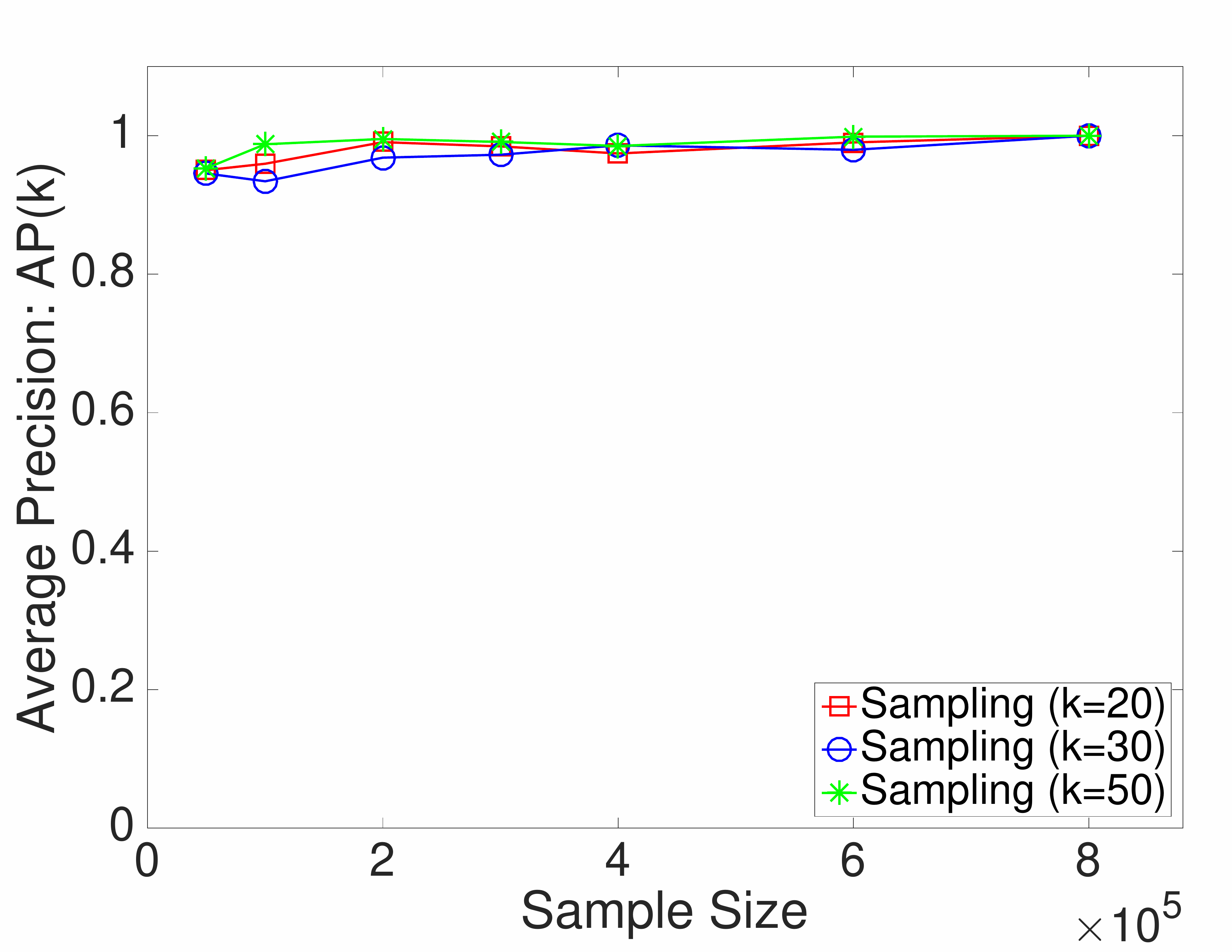}}
	\subfigure[AP, CN2, $l=2$]{\includegraphics[width=\parawidth]{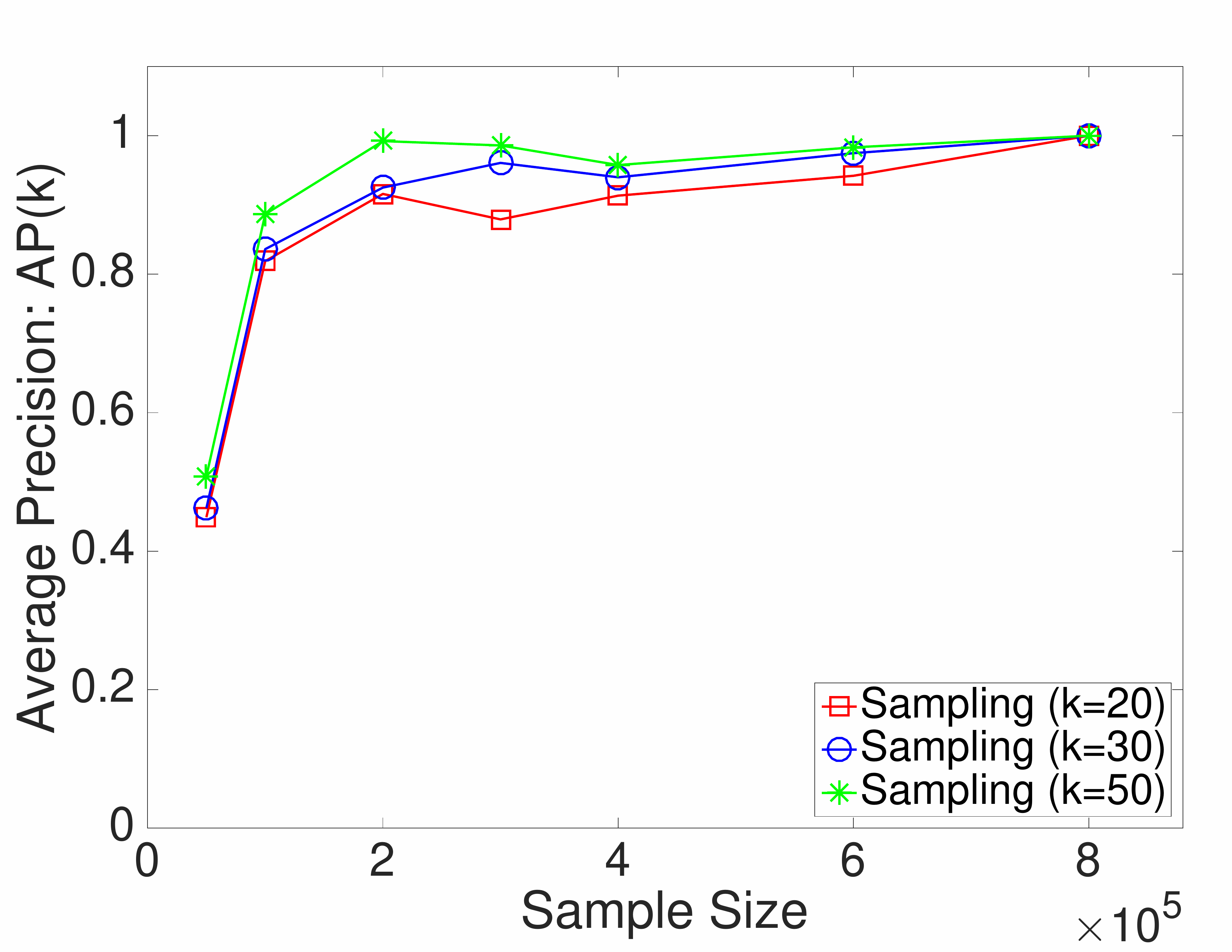}}
	
	\subfigure[RS, CN2, $l=1$]{\includegraphics[width=\parawidth]{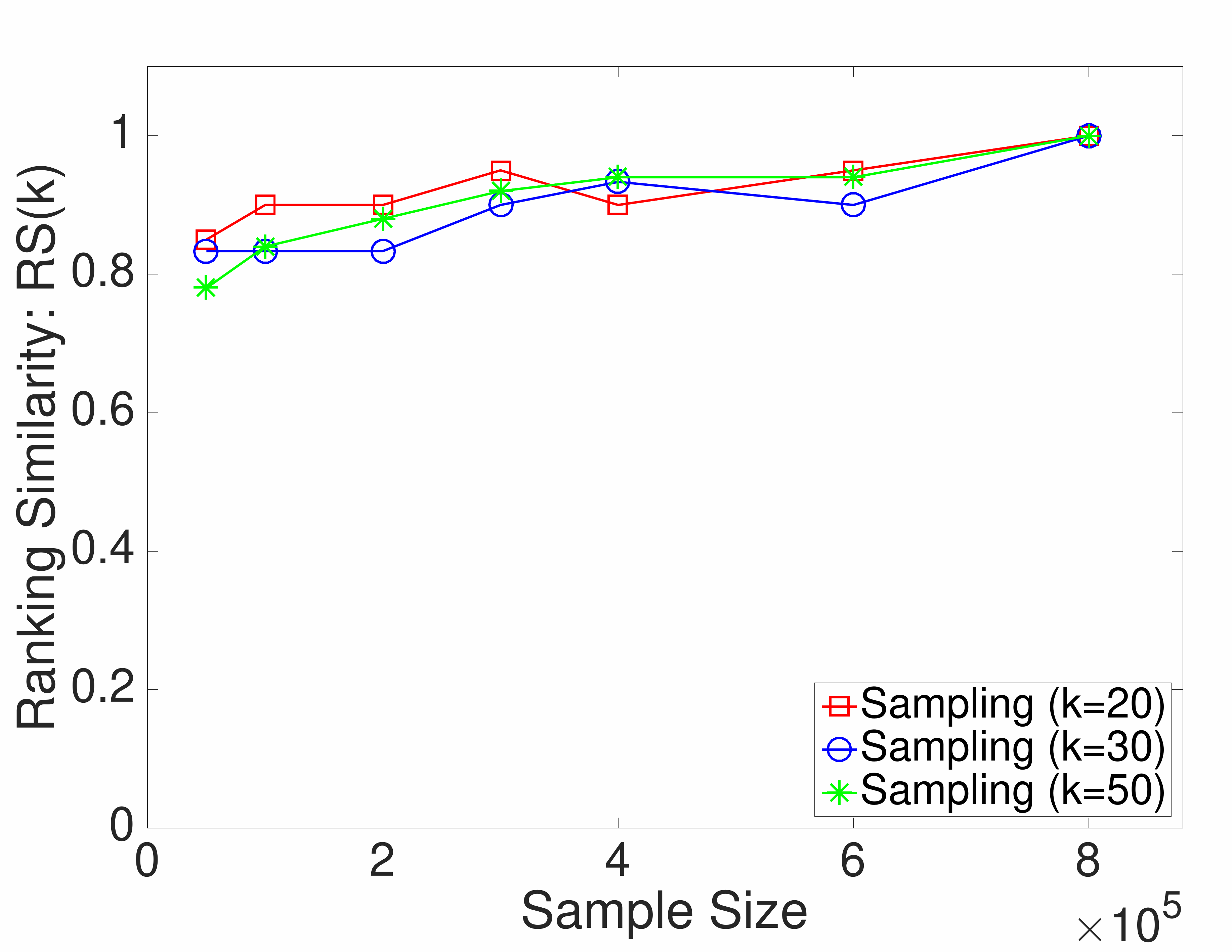}}
	\subfigure[RS, CN2, $l=2$]{\includegraphics[width=\parawidth]{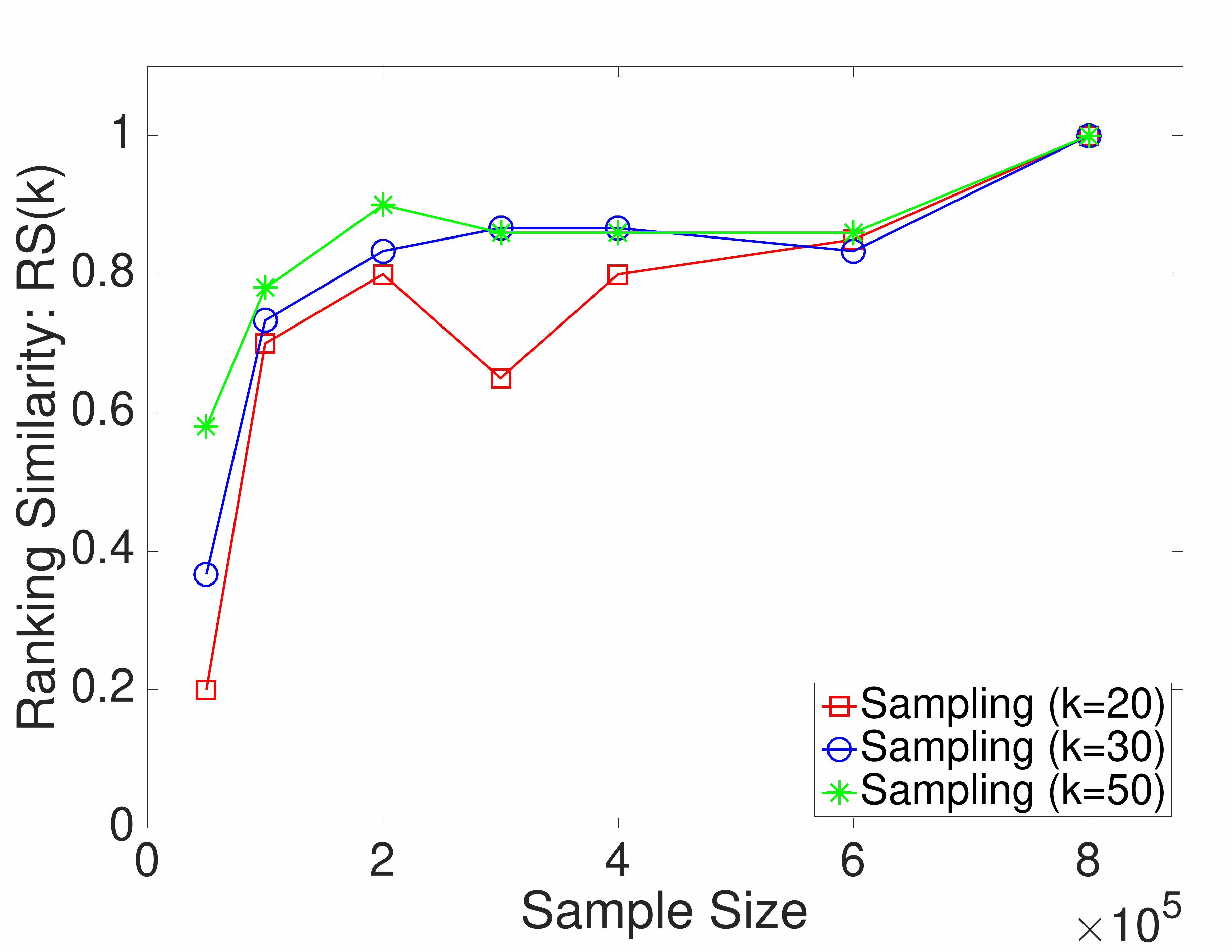}}
	
	\subfigure[RT, CN2, $l=1$]{\includegraphics[width=\parawidth]{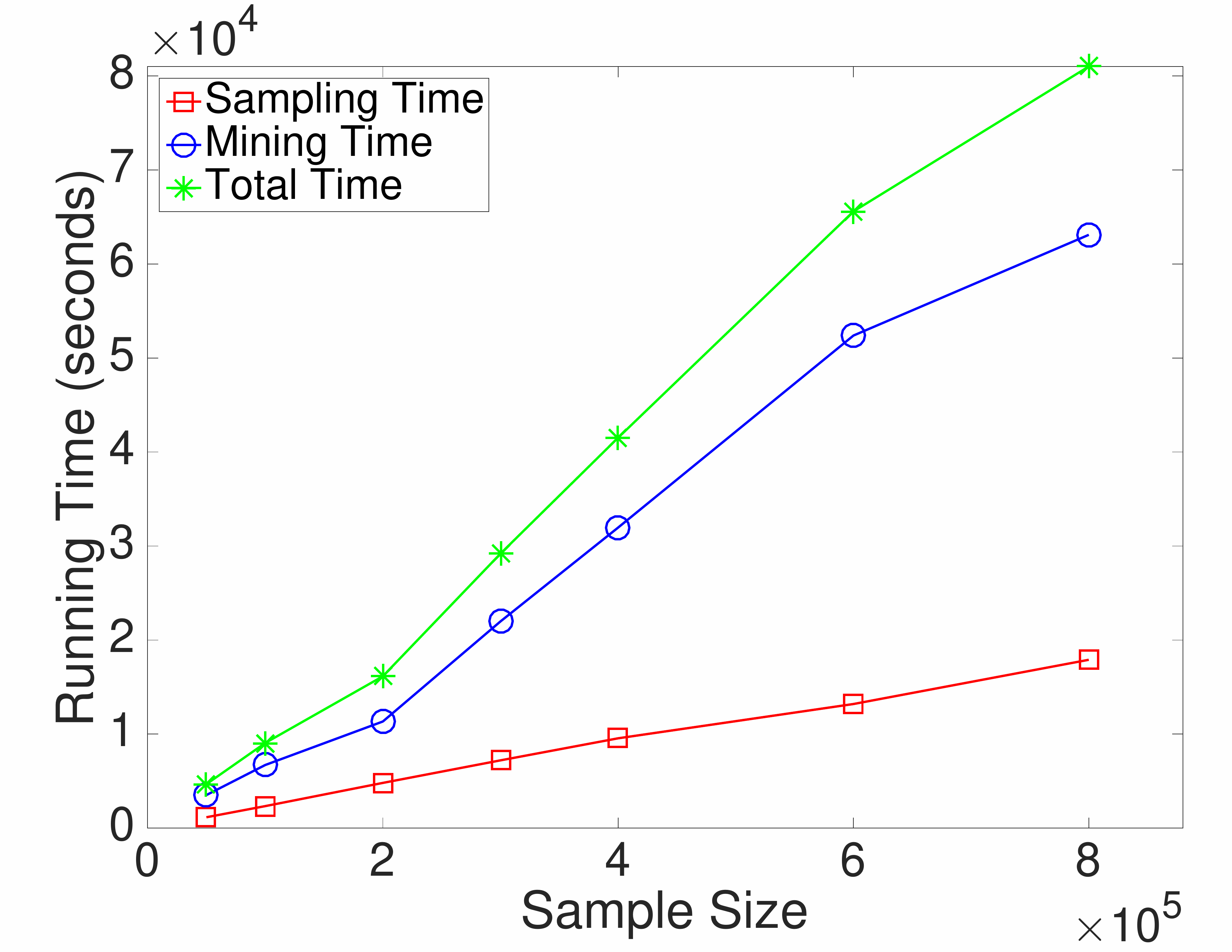}}	
	\subfigure[RT, CN2, $l=2$]{\includegraphics[width=\parawidth]{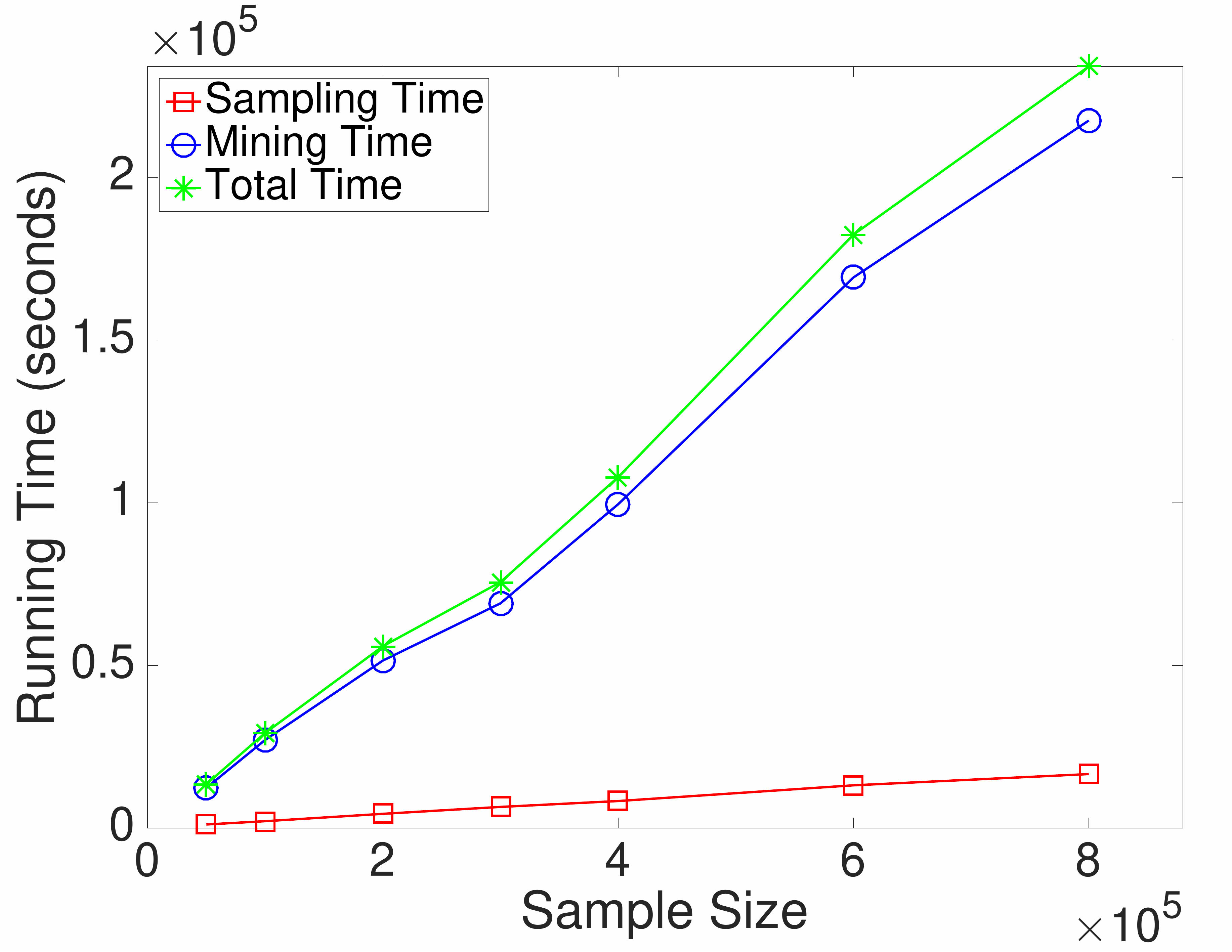}}
\caption{The AP, RS and RT performances of the sampling-based method on the CN2 dataset.}
\label{fig:exp_cn3} 
\end{figure}

\subsection{Results on Real-world Datasets}
In this section, we introduce the experiment results on the real-world datasets. We focus on analyzing how the AP, RS and runtime of the baseline and sampling-based methods change when the sample size $|S_l|$ increases. Since the baseline method cannot complete in a practical amount of time on the large datasets, we use the small real-world dataset Flight to demonstrate the performances of baseline and sampling-based methods, and test the scalability of our method on CN1 and CN2.

We first introduce the performance on the small dataset Flight, where the ground truth ranked list of patterns can be obtained by running the exact baseline method. For length $l=1$ and $l=2$, the numbers of all transaction sequences in Flight are 44,013 and 2,369,477, respectively. The sample sizes are set to $|S_1|=5\times 10^3, 1\times 10^4, 1.5\times 10^4, 2\times 10^4, 2.5\times 10^4, 3\times 10^4$ and $|S_2|=5\times 10^4, 1\times 10^5, 1.5\times 10^5, 2\times 10^5, 3\times 10^5, 5\times 10^5$ and $k=20,50,100,200$.

Figures~\ref{fig:flight}(a)-(d) show the AP and RS of the baseline method and the sampling-based method on Flight.
The baseline always has AP and RS of $1.0$ since it produces the exact results.
For the sampling-based method, when the sample size increases, the estimation error of pattern frequency decreases. 
As a result, when the sample size is not too small, the sampling-based method can produce a stable and accurate ranked list of patterns, which is comparable with the baseline in AP and RS.

Notice that, the AP and RS don't converge fast on Flight. When $k$ is large, the performances are not always good. The reason is that, the patterns in the long tail always have similar frequencies. This leads to that, if we want to distinguish them, we need much more samples, which may exceed the total number of transaction sequences. In other words, as we have explained above, for decreasing the estimation error of pattern frequency, we must increase the sample size.

Figures~\ref{fig:flight}(e)-(f) show the runtime of the baseline method and the sampling-based method. Since the baseline method is independent from the sample size of transaction sequences, the runtime of the baseline method does not change when the sample size increases. For the sampling-based method, since a larger sample size leads to a larger time cost in transaction sequence sampling and sequential pattern mining, the runtime of the sampling-based method increases when the sample size increases.  We will test more on the scalability on the larger citation network data sets next. 

We now introduce the experiment results on the CN1 and CN2 datasets. The sample sizes are set to $|S_l|=5\times 10^4, 1\times 10^5, 2\times 10^5, 3\times 10^5, 4\times 10^5, 6\times 10^5, 8\times 10^5$. To test the stability of the sampling-based method, we use the ranked list $\mathcal{G}'$ generated by the sampling-based method using sample size $|S_l|=8\times 10^5$ as the pseudo ground truth to evaluate the AP and RS of the sampling-based method using different sample sizes.  %Although such AP and RS performances are relative performances with respect to $\mathcal{G}'$, they are still valid metrics to analyze how the AP and RS performances of the sampling-based method change when the sample size increases from $|S_l|=5\times 10^4$ to $|S_l|=8\times 10^5$.

Figures~\ref{fig:exp_cn1}(a)-(b) and Figures~\ref{fig:exp_cn3}(a)-(b) show the AP of the sampling-based method on CN1 and CN2, respectively.
For both $l=1$ and $l=2$, when the sample size increases, the AP performance approaches $1$ quickly.
This demonstrates that the estimated ranked list of patterns becomes more stable and accurate when the sample size increases.
The AP for $l=1$ converges faster than the AP for $l=2$. The reason is that increasing $l$ from $l=1$ to $l=2$ increases the bound of the sample complexity(See Theorem~\ref{eq:lemma}). Thus, for the same sample size $|S_l|$, the AP with $l=1$ is better than that with $l=2$.
\nop{The reason is that increasing $l$ from $l=1$ to $l=2$ increases the bound of the sample size $|S_l|$ (i.e., $|S_l|\geq \frac{12|I|(l+1)+12}{\varepsilon^2 a}  ln\frac{2}{\delta}$ in Theorem~\ref{eq:lemma})}\nop{(Since the bound is a maximum value, we can not say that here. The reason is actually the variation of $a$, $\varepsilon$ and $\delta$)}

By comparing Figure~\ref{fig:exp_cn1}(b) and Figure~\ref{fig:exp_cn3}(b), we can see that the AP of the sampling-based method converges faster on CN1 than on CN2. 
CN2 is much larger than CN1 and CN1 is a subset of CN2. When the size of the data set increases, the total number of transaction sequences increases faster than the number of sequences that share a specific pattern. Thus, the pattern frequencies of highly ranked patterns in CN2 are smaller than those in CN1. Consequently, to keep the estimated ranked list of patterns stable and accurate, the error bound $\varepsilon$ for CN2 needs to be smaller than that for CN1.
As a result, we need a larger sample size in CN2 to achieve the same level of result quality as in CN1.

Figures~\ref{fig:exp_cn1}(c)-(d) and~\ref{fig:exp_cn3}(c)-(d) show the RS of the sampling-based method.
Similar to the trend of AP in Figures~\ref{fig:exp_cn1}(a)-(b) and~\ref{fig:exp_cn3}(a)-(b), for $l=1$ and $l=2$, the RS of the sampling-based method converges fast to the best results as well. Due to the same reasons in the situation of AP, the RS for $l=1$ converges faster than that for $l=2$ and the RS on CN1 converges faster than that on CN2.

Figures~\ref{fig:exp_cn1}(e)-(f) and Figures~\ref{fig:exp_cn3}(e)-(f) show the runtime of the sampling-based method on CN1 and CN2. 
The \emph{sampling time} is the time cost of sampling the set of transaction sequences $S_l$. 
The \emph{mining time} is the time cost of mining top-$k$ sequential patterns from $S_l$.
The \emph{total time} is the overall time cost of the sampling-based method, which is the sum of the sampling time and the mining time. 
For both $l=1$ and $l=2$, the runtime of the sampling-based method increases almost linearly when the sample size increases. This demonstrates the superior scalability of the proposed sampling-based method.

In summary, the sampling-based method can obtain stable top-$k$ sequential patterns when the sample size is not too small, and is much more scalable than the baseline method in mining frequent sequential patterns from large graphs.

\begin{table}[t]
\caption{Some interesting patterns in CN1 ($l=1$)}
\centering
\label{tab:casePattern_1}
	\begin{tabular}{|p{6.9cm}|r|}
	\hline
	\centering  Sequential Patterns & \#RP\\
	\hline
	$\langle$(machine learning), (social network)$\rangle$ & 21,200\\
	\hline
	$\langle$(social network), (random walk)$\rangle$ & 12,900\\
	\hline
	$\langle$(deep learning), (clustering)$\rangle$ & 12,600\\
	\hline
	$\langle$(graph embedding), (classification)$\rangle$ & 5,390\\
	\hline
	$\langle$(social network), (anomaly detection)$\rangle$ & 4,920\\
	\hline		
	$\langle$(deep learning), (reinforcement learning)$\rangle$  & 3,870\\				
	\hline
	$\langle$(graph partitioning), (community detection)$\rangle$ & 3,500\\
	\hline
	$\langle$(dynamic network), (game theory)$\rangle$ & 3,340\\
	\hline
	$\langle$(deep learning), (anomaly detection)$\rangle$ & 1,430\\	
	\hline
	$\langle$(network evolution), (game theory)$\rangle$ & 1,170\\
	\hline
	\end{tabular}
\end{table}

\begin{table}[t]
\caption{Some interesting patterns in CN1 ($l=2$)}
\centering
\label{tab:casePattern_2}
	\begin{tabular}{|p{6.9cm}|r|}
	\hline
	\centering Sequential Patterns & \#RP\\
	\hline
	$\langle$(data mining), (machine learning), (information retrieval)$\rangle$ & 37,700\\
	\hline
	$\langle$(machine learning), (social network), (approximation algorithm)$\rangle$ & 10,900\\
	\hline
	$\langle$(data mining), (machine learning, social network), (web search)$\rangle$ & 7,450\\
	\hline
	$\langle$(deep learning), (clustering), (prediction)$\rangle$ & 7,430\\
	\hline	
	$\langle$(dynamic network), (game theory), (clustering)$\rangle$ & 4,410\\
	\hline
	$\langle$(social network), (random walk), (recommendation)$\rangle$ & 4,180\\
	\hline
	$\langle$(social network), (anomaly detection), (classification)$\rangle$  & 3,240\\
	\hline
	$\langle$(graph embedding), (classification), (clustering)$\rangle$ & 2,980\\
	\hline
	$\langle$(deep learning), (reinforcement learning), (classification)$\rangle$ & 2,660\\	
	\hline
	$\langle$(social network), (anomaly detection), (pattern mining)$\rangle$ & 518\\
	\hline
	\end{tabular}
\end{table}

\subsection{A Case Study}

In this section, we report a case study on the CN1 dataset.
Mining frequent sequential patterns in CN1 finds interesting patterns of topic relations in paper citations. 

Tables~\ref{tab:casePattern_1}-\ref{tab:casePattern_2} show some interesting patterns obtained by the sampling-based method with $l=1$ and $l=2$, respectively. 
Those frequent patterns show significant trends of cross-field citations, which reveal active interdisciplinary research hotspots among different research fields.
The activeness of those research hotspots is also confirmed by \emph{the number of related publications} (\#RP) obtained by searching \emph{Google Scholar} (\emph{scholar.google.com}) using the keywords in the mined sequential patterns. 
Specifically, the patterns in Tables~\ref{tab:casePattern_1}-\ref{tab:casePattern_2} are sorted by the estimated pattern frequencies obtained by our sampling-based method. At the same time, the ranked lists are highly consistent with descending order in \#RP value. This consistency demonstrates that the results produced by the sampling-based method matches the real world patterns, and verifies the effectiveness of mining sequential patterns in database graphs.

%To further investigate the underlying meaning of each frequent pattern, we extract all the paper citations related to each pattern and study the internal reasons why such frequent patterns are frequent.

We also look into the specific patterns and find they are meaningful. Take $\langle$(deep learning), (reinforcement learning)$\rangle$ as an example. By extracting from Google Scholar all paper citations related to this pattern, we find that many deep learning papers cite reinforcement learning papers in recent years.  As we know, deep reinforcement learning has been one of the hottest research trends in the field of deep learning.
%which gave birth to the famous \emph{Alpha Go} and has been widely used in Recommendation Systems, Dialogue Management and Text-based Games.

The relationship between machine learning and social network, captured by some patterns found, is also interesting. In Table~\ref{tab:casePattern_1}, the top pattern $\langle$(machine learning), (social network)$\rangle$ reveals the hot research trend of applying machine learning algorithms to solve social network problems.
%, such as network intrusion detection, classification, link prediction, etc.
However, it is an open problem that classic machine learning algorithms are not scalable enough to handle large scale social networks, thus many researchers use approximation techniques to improve the scalability of machine learning algorithms on social networks.
This trend is captured by pattern $\langle$(machine learning), (social network), (approximation algorithm)$\rangle$ in Table~\ref{tab:casePattern_2}.

\section{Conclusion}
\label{sec:con}
In this paper, we tackled the novel problem of finding top-$k$ sequential patterns in database graphs. We designed a fast sampling-based top-$k$ sequential pattern mining algorithm. Our experiments on both synthetic data sets and real data sets showed the superior effectiveness and efficiency of the proposed sampling method in finding meaningful sequential patterns. As future work, we will further improve the scalability and extend the sampling method to handle dynamic graphs.

%\begin{acknowledgements}
%\end{acknowledgements}

\bibliographystyle{spbasic}
\bibliography{reference}

\end{document}